\documentclass[a4paper,11pt]{amsart}
\usepackage{amsfonts,amssymb,epsfig,latexsym}
\usepackage{amsmath}
\usepackage[latin1]{inputenc}
\usepackage{color,layout}
\usepackage{ifthen}

\usepackage{pdfsync}
\usepackage{graphicx}
\usepackage{color}
\usepackage{pdfsync}
\date{}

\newcommand{\be}{\begin{equation}}
\newcommand{\ee}{\end{equation}}

\def\la{\langle}
\def\ra{\rangle}

\def\R{\mathbb{R}}
\def\C{\mathbb{C}}
\def\Z{\mathbb{Z}}
\def\N{\mathbb{N}}

\renewcommand{\Re}{\text{{\rm Re}\;}}
\renewcommand{\Im}{\text{{\rm Im}\;}}

\newcommand{\supp}{\text{{\rm supp}}}

\newcommand{\dist}{\text{\rm dist}}

\newtheorem{theorem}{Theorem}[section]
\newtheorem{lemma}[theorem]{Lemma}
\newtheorem{proposition}[theorem]{Proposition}
\newtheorem{corollary}[theorem]{Corollary}

\theoremstyle{definition}

\newtheorem{remark}[theorem]{Remark}

\numberwithin{equation}{section}

\title{Resonance widths for the molecular predissociation}
\author{
Alain GRIGIS${}^1$ \& 
Andr\'e MARTINEZ$ {}^2$
 }

\begin{document}

\maketitle 
\addtocounter{footnote}{1}
\footnotetext{{\tt\small Universit\'{e} Paris 13,  
D\'epartement de Math\'ematiques,  LAGA UMR CNRS 7539, Av. J.-B. Cl\'ement, 93430 Villetaneuse, France, grigis@math.univ-paris13.fr}  \\Partly supported by the GDR  INdAM-CNRS GrefiMefi}  
\addtocounter{footnote}{1}
\footnotetext{{\tt\small Universit\`a di Bologna,  
Dipartimento di Matematica, Piazza di Porta San Donato, 40127
Bologna, Italy, 
andre.martinez@unibo.it }\\Partly supported by Universit\`a di Bologna, Funds for Selected Research Topics}  
\begin{abstract}
We consider a semiclassical $2\times 2$ matrix Schr\"odinger operator of the form $P=-h^2\Delta {\bf I}_2 +
{\rm diag}(V_1(x) ,  V_2(x)) +hR(x,hD_x)$, where  $V_1, V_2$ are real-analytic, $V_2$ admits a non degenerate minimum at 0, $V_1$ is non trapping at energy $V_2(0)=0$, and $R(x,hD_x)=(r_{j,k}(x,hD_x))_{1\leq j,k\leq 2}$ is a symmetric off-diagonal $2\times 2$ matrix of first-order pseudodifferential operators with analytic symbols. We also assume that $V_1(0) >0$. Then, denoting by $e_1$ the first eigenvalue of $-\Delta +  \la  V_2''(0)x,x\ra /2$, and under some ellipticity condition on $r_{1,2}$ and additional generic geometric assumptions, we show that  the unique resonance $\rho_1$ of $P$ such that $\rho_1 = V_2(0) + (e_1+r_{2,2}(0,0))h + {\mathcal O}(h^2)$ (as $h\rightarrow 0_+$) satisfies, 
$$
\Im \rho_1 = -h^{n_0+(1-n_\Gamma)/2}f(h,\ln\frac1{h})e^{-2S/h},
$$
where $f(h,\ln\frac1{h}) \sim \sum_{0\leq m\leq\ell} f_{\ell,m}h^\ell(\ln\frac1{h})^m$ is a symbol with $f_{0,0}>0$, $S>0$ is the so-called Agmon distance associated with the degenerate metric $\max(0, \min(V_1,V_2))dx^2$, between 0 and $\{V_1\leq 0\}$, and $n_0\geq 1$, $n_{\Gamma}\geq 0$ are integers that depend on the geometry.
\end{abstract}  

\baselineskip = 18pt 
\vfill\eject
\section{Introduction}
The theory of predissociation goes back to the very first years of quantum mechanics (see, e.g., \cite{Kr, La, Ze, St}). 
Rouhgly speaking, it describes the possibility for a molecule to dissociate spontaneously (after a sufficiently large time) into several sub-molecules, for energies below the crossing of the corresponding energy surfaces of the initial molecule and the final dissociated state. From a physical point of view, one naturally expects that this (typically quantic) phenomenon occurs with extremely small (but non zero) probability.

Despite the fact that statements concerning this problem are present in the physics literature for more than 70 years, the first mathematically rigorous result is due to M. Klein \cite{Kl} in 1987, where an upper bound on the time of predissociation is given in the framework of the Born-Oppenheimer approximation. More precisely, denoting by $h$ the square root of the ratio of electronic to nuclear mass, M. Klein proves the existence of resonances $\rho$ with real part below the crossing of the energy surfaces, and with exponentially small imaginary part, that is,
$$
|\Im\rho | ={\mathcal O}(e^{-2(1-\varepsilon)S/h})
$$
where $S>0$ is a geometric constant, $\varepsilon >0$ is fixed arbitrarily, and the estimate holds uniformly as $h$ goes to zero.

In terms of probabilities, this result corresponds to give an upper bound on the transition probability between the initial molecule and the dissociated state. The purpose of this article is to obtain a more complete information on this quantity, and in particular a lower bound on it. More precisely, under suitable conditions, we prove that the imaginary part of the lowest resonance admits a complete asymptotic expansion of the type,
$$
\Im \rho_1 = -h^{n_0+(1-n_\Gamma)/2}e^{-2S/h}\sum_{0\leq m\leq\ell} f_{\ell,m}h^\ell(\ln\frac1{h})^m,
$$
in the sense that, for any $N\geq 1$, one has 
\begin{eqnarray*}
|\Im \rho_1 +h^{n_0+(1-n_\Gamma)/2}e^{-2S/h}\sum_{0\leq m\leq\ell\leq N} f_{\ell,m}h^\ell(\ln\frac1{h})^m| \\
={\mathcal O}(h^{n_0+(1-n_\Gamma)/2+N}e^{-2S/h}),
\end{eqnarray*}
where $S>0$, $n_0\geq 1$ and $n_{\Gamma}\geq 0$ are all geometric constants,
and where the leading coefficient $f_{0,0}$ is positive.
\vskip 0.3cm
As it is well-known, the quantity $\Im\rho$ is closely related to the oscillatory behavior of the corresponding resonant state in the unbounded classically allowed region. Hence, the main issue will be to know sufficiently well this behavior.
\vskip 0.3cm
The strategy of the proof consists in starting from the WKB construction  at the bottom of the well, and then trying to extend them as much as possible, and at least up to the classically allowed unbounded region. This is mainly the same strategy used in \cite{HeSj2} for the study of shape resonances.
\vskip 0.3cm

However, from a technical point of view, several new problems are encountered, because of the crossing of the electronic levels. 

\vskip 0.3cm
The first one is that, at the crossing, the only reference on WKB constructions is that of \cite{Pe}, that has been done for a special type of matrix Schr\"odinger operators. In particular, it strongly uses  the fact that only differential operators are involved. In our case, since our operator comes from a Born-Oppenheimer reduction, it is necessarily of pseudodifferential kind (see, e.g., \cite{KMSW, MaSo}). As a consequence, our first step will consist in extending the method of \cite{Pe} to pseudodifferential operators. Unfortunately, this extension is far from being straightforward, and needs a specific formal calculus adapted to expressions involving the Weber functions.

\vskip 0.3cm
The second one is that, after having overcome the crossing, the symbols of the resulting WKB expansions do not anymore satisfy analytic estimates (usually needed in order to re-sum them, up to exponentially small error terms). In particular, this prevents us from using directly the constructions of \cite{HeSj2} near the classically allowed unbounded region. Instead, we have to adapt the method of \cite{FLM} that, without analyticity, allows us to extend the WKB constructions into the classically allowed unbounded region up to a distance of order $(h\ln|h|)^{2/3}$ from the barrier. This is not much, but this is enough for having a sufficient control, in this region, on the difference between the true solution and the WKB one. This is actually done by adapting the specific arguments of propagation introduced in \cite{FLM}, where the propagation takes place in $h$-dependent domains.

\vskip 0.3cm
In the next section, we describe in details the geometrical context and the assumptions. 

In Section \ref{secres}, we state our main result. 

Section \ref{constWKB} is devoted to the WKB constructions, starting from the well and proceeding away along some minimal geodesics, until crossing the boundary of the classically  forbidden region. It is in this section that we develop a formal pseudodifferential calculus adapted to expressions involving the Weber functions. 

Next, in Section \ref{secagm}, we extend the well-known Agmon estimates to our pseudodifferential context. In this case, the main feature is that, since we cannot use general Lipschitz weight-functions, we replace them by $h$-dependent smooth functions with bounded gradient, but with derivatives of higher order that can grow to infinity as $h\to 0$. 

In Section \ref{secglob}, we use these estimates in order to obtain a bound for the difference between the WKB solutions and a solution of a modified problem, and this permits us to define an asymptotic solution in a whole neighborhood of the classically forbidden region  (but only up to a distance of order $(h\ln|h|)^{1/3}$ from this region). 

Section \ref{compsol} contains the {\it a priori} estimates and the propagation arguments that lead to a good control on the difference between the asymptotic solution and the actual one.

Finally, Section \ref{secwidt} makes the link with the width of the resonance. Even if the idea is standard (practically an application of the Green  formula: see, e.g., \cite{HeSj2}), here we have to be careful with the double problem  that, on the one hand, we deal with pseudodifferential (not differential) operators, and, on the other hand, the magnitude of freedom outside the classically forbidden region is of order $(h\ln|h|)^{1/3}$ as $h\to 0$.

\section{Geometrical Assumptions}

We consider the semiclassical $2\times 2$ matrix Schr\"odinger operator,
\be
\label{operator}
P= 
\left(\begin{array}{cc}
P_1 & 0\\
0 & P_2
\end{array}\right) + hR(x,hD_x)
\ee
with,
$$
P_j := -h^2\Delta +V_j(x) \quad (j=1,2),
$$
where $x=(x_1,\dots ,x_n)$ is the current variable in $\R^n$ ($n\geq 1$),
$h>0$ denotes the semiclassical parameter, and $R(x,hD_x)=(r_{j,k}(x,hD_x))_{1\leq j,k\leq 2}$ is a  formally self-adjoint  $2\times 2$ matrix of first-order
semiclassical pseudodifferential operators. 

Let us observe that this is typically the kind of operator one obtains in the Born-Oppenheimer approximation, after reduction to an effective Hamiltonian (see \cite{KMSW, MaSo}). In that case, the quantity $h^2$ stands for the inverse of the mass of the nuclei.

\vskip 0.2cm

{\bf Assumption 1.} { \it The potentials $V_1$ and $V_2$ are smooth and 
bounded on $\R^n$,   and
satisfy,
\begin{eqnarray}
\label{assV1}
&& V_1(0) >0\mbox{ and } E=0 \mbox{ is a non-trapping energy for } V_1; \\
&&V_1 \; \mbox{has a strictly negative limit as} \; \vert x\vert \rightarrow \infty;\\
\label{assV2}
&&V_2\geq 0\; ; \; V_2^{-1}(0) =\{ 0\} \; ; \;  {\rm Hess}V_2(0) >0 \; ; \; 
\liminf_{\vert x\vert\rightarrow\infty} V_2 >0.
\end{eqnarray}
}
\begin{center}  
{\includegraphics[scale=0.40]{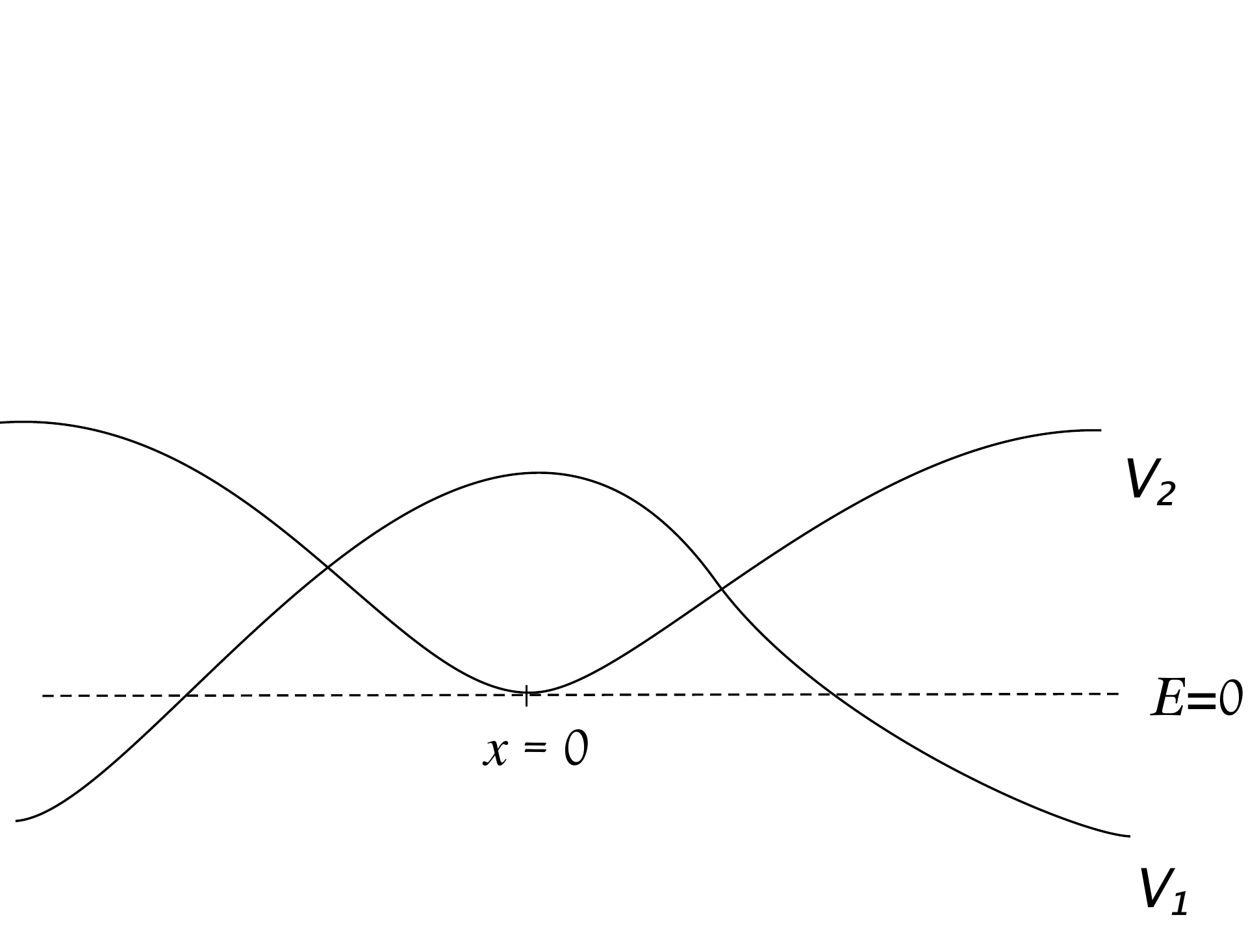}}\end{center}
\vskip 1cm

In particular we assume that $V_2$ has a unique non degenerate well at $x=0$.  This well is included in the island  $ \ddot O$ which is the bounded open set defined as,

\be
\label{island}
{\ddot O} = \lbrace x \in \R^n; V_1 (x) >0 \rbrace.
\ee

Next, we define the sea as the set where $V_1(x) <0$.

The fact that 0 is a non-trapping energy for $V_1$ means that, for any $(x,\xi)\in p_1^{-1}(0)$, one has $|\exp tH_{p_1}(x,\xi)|\rightarrow +\infty$ as $t\rightarrow \infty$, where $p_1(x,\xi):=\xi^2+V_1(x)$ is the symbol of $P_1$, and $H_{p_1}:=(\nabla_\xi p_1, -\nabla_x p_1)$ is the Hamilton field of $p_1$.
\vskip 0.2cm
Such conditions (\ref{assV1})-(\ref{assV2}) correspond to the situation of {\it molecular predissociation} as described in  \cite{Kl}.
\vskip 0.2cm

Since we plan to study the resonances of $P$ near the energy level $E=0$, we also assume,

{\bf Assumption 2.}\label{Ass2} {\it The potentials $V_1$ and $V_2$ extend to bounded holomorphic functions near a complex sector of the form,
${\mathcal S}_{R_0,\delta} :=\{x\in \C^n\, ;\, |\Re x|\geq R_0\, ,\,  \vert \Im x\vert \leq \delta |\Re x| \}$, with $R_0,\delta >0$. Moreover  $V_1$ tends to its limit at $\infty$ in this sector and $\Re V_2 $ stays away from $0$ in this sector.}

{\bf Assumption 3.} \label{Ass3}{\it The  symbols $r_{j,k}(x,\xi)$ for $(j,k) = (1,1), (1,2), (2,2)$ extend to  holomorphic functions near,
$$
\widetilde {\mathcal S}_{R_0,\delta}:= {\mathcal S}_{R_0,\delta}\times \{ \xi\in \C^n\,; |\Im\xi| \leq \max (\delta\la \Re x \ra, \sqrt{M_0})\},
$$
with,
$$
M_0> \sup_{x\in\R^n}\min (V_1(x), V_2(x)).
$$
and, for any $\alpha\in \N^{2n}$,  they satisfy
\be
\label{assR}
\partial^\alpha r_{j,k}(x,\xi) ={\mathcal O}(\la \Re \xi\ra)\,\, \mbox{ uniformly on } \widetilde {\mathcal S}_{R_0,\delta}.
\ee}

Now we define the cirque $\Omega$ as,

\be
\label{cirque}
\Omega = \lbrace x \in \R^n; V_2 (x) < V_1 (x) \rbrace
\ee

Hence, the well is in the cirque and the cirque is in the island.

 We also consider the Agmon distance associated to the pseudo-metric
 $$(\min(V_1, V_2))_+ dx^2.$$
 Such a metric is considered in Pettersson \cite{Pe}.  There are three places where this metric is not a standard one. 

First, near the well $0$, but this case is well known. It has been treated by Helffer and Sj\" ostrand \cite{HeSj1}. It is also considered in Pettersson. The Agmon distance,
\be
\label{Agmon distance}
\varphi(x) = d(x,0)
\ee
is smooth at $0$.
The point $(x,\xi) = (0,0)$ is a hyperbolic singular point of the Hamilton vector field $H_{q_2}$, where $q_2= \xi^2 - V_2(x)$, and the stable and unstable manifold near this point are respectively the Lagrangian manifolds $\lbrace \xi = \nabla \varphi(x) \rbrace$ and $\lbrace \xi = - \nabla \varphi(x) \rbrace$.

Secondly on $\partial \Omega$, precisely at the points where $V_1 = V_2$. This case has been also considered in Pettersson.  At such a point, if one assume that $\nabla V_1 \not = \nabla V_2$, then, any geodesic which is tranversal to the hypersurface $\lbrace V_1= V_2 \rbrace$ is ${C}^1$.

Finally there is the boundary of the island $\partial \ddot O$, where $V_1=0$. This situation was considered in Helffer-Sj\"ostrand \cite{HeSj2}. We will follow them in a next assumption.

Now we consider the distance from the well to the sea, that means to $\partial \ddot O$~:
\be
\label{S}
S = d(0,\partial \ddot O),
\ee

Setting $B_{S}= \lbrace x \in \ddot O, \varphi( x) <S\rbrace$ and denoting by $\bar {B}_{S}$ its closure, we also consider  the set $\bar {B}_{S} \cap \partial \ddot O$, that consists in the  points of the boundary of the island that are joined to the well by a minimal $d$-geodesic included in the island. These points are called points of type 1 in \cite{HeSj2}, and we denote by $G$ the set of the minimal geodesics joining such a point to $0$ in $\ddot O$.
\vskip 0.3cm
We make the following assumption.
\vskip 0.3cm
{\bf Assumption 4.} \label{Ass4}{\it For all $\gamma\in G$,  $\gamma$ intersects $\partial \Omega$  at a finite number of points and the intersection is transversal at each of these points. Moreover,  $\nabla V_1 \not= \nabla V_2$ on $\gamma\cap \partial \Omega$.}

\vskip 0.3cm
Let us recall that the assumption that $0$ is a non trapping energy for $V_1$ implies that $\nabla V_1  \not = 0$ on $\partial \ddot O$, and therefore that $\partial \ddot O$ is a smooth hypersurface.
\vskip 0.3cm
We define the caustic set $\mathcal C$ as the closure of the set of points $x \in {\ddot O}$ with $\varphi(x) = S + d(x, \partial \ddot O)$. In particular, the points of type 1 are in $\mathcal C$. As in \cite{HeSj2} we assume,
\vskip 0.3cm
{\bf Assumption 5.} 
\label{Ass5}{\it
The points of type 1 form a submanifold $\Gamma$,
and $\mathcal C$ has a contact of order exactly two with $\partial \ddot O$ along $\Gamma$.}
\vskip 0.3cm
We denote by $n_{\Gamma}$ the dimension of $\Gamma$.  Moreover, for any $\gamma\in G$, we denote by $N_\gamma:= \# (\gamma\cap\partial\Omega)$ the number of points where $\gamma$ crosses the boundary of the cirque,
and we set,
$$
n_0:= \min_{\gamma\in G} N_\gamma\quad ;\quad G_0:=\{ \gamma\in G\, ;\, N_\gamma = n_0\}.
$$
Then, we make an assumption that somehow  insures that an interaction between the two Schr\" odinger operators does exist. 

{\bf Assumption 6.} \label{Ass6}{\it
There exists at least one $\gamma\in G_0$ for which one has the ellipticity condition  $r_{12} (x, i\nabla \varphi (x)) \not = 0$ at every  point $x\in \gamma\cap\partial \Omega$.}


\section{Main Result}
\label{secres}

Under the previous assumption we plan to study the resonances of the operator $P$ given in (\ref{operator}), where $R(x, hD_x)$ is defined as
$$
R(x,hD_x):= \left( \begin{array}{cc}
{\rm Op}^L_h (r_{1,1})&{\rm Op}^L_h (r_{1,2})\\
{\rm Op}^R_h (\overline{r_{1,2}})& {\rm Op}^L_h (r_{2,2})
\end{array} \right)
$$
where for any symbol $a(x,\xi)$ we use the following quantizations,
\begin{eqnarray*}
&& {\rm Op}^L_h (a)u(x) = \frac{1}{(2 \pi h)^n} \int e^{i(x-y)\xi/h} a(x, \xi) u(y) dy d\xi;\\
&&{\rm Op}^R_h (a)u(x) = \frac{1}{(2 \pi h)^n} \int e^{i(x-y)\xi/h} a(y, \xi) u(y) dy d\xi.
\end{eqnarray*}

In order to define the resonances we consider the distortion given as follows: Let $F(x) \in C^\infty (\R^n, \R^n)$ such that $F(x) = 0$ for $\vert x \vert \leq R_0$, $F(x) = x$ for $\vert x\vert$ large enough. For $\theta >0$ small enough, we define the distorded operator $P_{\theta}$ as the value at $ \nu = i \theta$ of  the extension to the complex of the  operator  $U_\nu P U_\nu^{-1}$ which is  defined for $\nu$ real small enough, and analytic in $\nu$, where we have set
\be
U_\nu \phi(x) = \det ( 1 + \nu dF(x))^{1/2} \phi ( x + \nu F(x)).
\ee 
Since we have a pseudodifferential operator $R(x, hD)$ the fact that  $U_\nu P U_\nu^{-1}$ is analytic in $\nu$ is not completely standard but can be done without problem (thanks to Assumption 3), and by using the Weyl Perturbation Theorem, one can also see  that there exists $\varepsilon_0 >0$ such that for any $\theta >0$ small enough, the spectrum of $P_\theta$ is discrete in $[ -\varepsilon_0, \varepsilon_0 ] - i [0, \varepsilon_0 \theta]$. The eigenvalues of $P_\theta$ are called the resonances of $P$ \cite{Hu, HeSj2, HeMa}. 

We will need another small parameter $k>0$ related to the semiclassical parameter $h>0$, defined as,
\be
\label{defk}
k := h \ln\frac1{h}.
\ee

In the sequel, we will study the resonances in the domain $[- \varepsilon_0, Ch] -i [0, Ck]$, where $C>0$ is arbitrarily large. In this case, we can adapt the WKB constructions near the well made in \cite{HeSj1}, and show that these resonances form a finite set $\{\rho_1,\dots, \rho_m\}$, with asymptotic expansions as $h \rightarrow 0$, of the form,
$$
\rho_j \sim h\sum_{\ell \geq 0}\rho_{j,\ell}h^{\ell/2},
$$
where $\rho_{j,\ell}\in \R$ and $\rho_{j,0}=e_j+r_{2,2}(0,0)$, $e_j$ being the $j$-th eigenvalue of the
harmonic oscillator $-\Delta + \la V_2''(0)x,x\ra /2$ (actually, to be more precise, one must also assume that the arbitrarily large constant $C$ does not coincide with one of the $e_j$'s).

In this paper we are interested in the imaginary part of these resonances. We have,

\begin{theorem}\sl
\label{mainth}
Under Assumptions 1 to 6,
 the  first resonance
$\rho_1$ of $P$ is such that,
$$
\Im \rho_1 = -h^{n_0 +(1-n_\Gamma)/2}f(h, \ln\frac1{h})e^{-2S/h}
$$
where  $f(h, \ln\frac1{h})$ admits an asymptotic expansion of the form,
$$
f(h, \ln\frac1{h})\sim \sum_{0\leq m\leq\ell}f_{\ell, m}h^\ell (\ln\frac1{h})^m,\quad  (h\rightarrow 0),
$$
with $f_{0,0} >0$, and  $S>0$ is defined in (\ref{S}).

Moreover the other resonances in $[- \varepsilon_0, Ch] -i [0, Ck]$ verify
$$
\Im \rho_j = {\mathcal O} (h^{\beta_j} e^{-2S/h}),
$$
for some real $\beta_j$, uniformly as $h \rightarrow 0$.
\end{theorem}

\section{WKB constructions}\label{constWKB}
In this chapter, we fix some minimal $d$-geodesic $\gamma\in G$, and we denote by $x^{(1)}, \dots, x^{(N_\gamma)}$ the sequence of points that constitute $\gamma\cap\partial\Omega$, ordered from the closest to $0$ up to the closest to $\ddot O$ (note that $N_\gamma$ is necessarily an odd number). We also denote by $\gamma^{(1)}, \gamma^{(2)}, \dots, \gamma^{(N_\gamma +1)}$ the portions of $\gamma\backslash \partial\Omega$ that are in-between  $0$ and $x^{(1)}$, $x^{(1)}$ and $x^{(2)}$, ..., $x^{(N_\gamma)}$ and $\ddot O$, respectively, in such a way that we have,
$$
\gamma = \gamma^{(1)} \cup \{x^{(1)}\} \cup \gamma^{(2)}\cup \dots \cup \{x^{(N_\gamma)}\}\cup \gamma^{(N_\gamma +1)},
$$
where the union is disjoint (in particular, by convention we assume that $0\in \gamma^{(1)}$). Moreover, we start by considering the first resonance $\rho_1$ only.
\subsection{In the cirque}
As in \cite{Pe}, the starting point of the
construction consists in the WKB asymptotics
given near the well $x=0$ by a method due to
Helffer and Sj\"ostrand \cite{HeSj1}. More precisely,
because of the matricial nature of the
operator and the fact that $p_1$ is
elliptic above $x=0$, one finds a formal
solution $w_1$ of $Pw_1=\rho_1 w_1$ of the form,
\be
\label{BKWinit}
w_1(x;h)=\left(\begin{array}{clcr} 
ha_1(x,h) \\ 
a_2(x,h)\end{array}\right) e^{-\varphi
(x)/h}.
\ee
where $\varphi$ is defined in (\ref{Agmon distance}), and  $a_j$ ($j=1,2$) is a classical symbol of order
$0$ in $h$, that is, a formal series
in $h$ of the form,
\be
\label {AE1}
a_j(x,h) = \sum_{k=0}^\infty h^ka_{j,k}(x)
\ee
with $a_{j,k}$ smooth near 0 (here no half-powers of $h$ appear since we consider the first resonance $\rho_1$ only). Moreover, $a_2$ is
elliptic in the sense that $a_{2,0}$ never
vanishes. Note that the generalization of the
constructions of \cite{HeSj1} to the case of
pseudodifferential operators is done by the
use of a so-called {\it formal semiclassical
pseudodifferential calculus},  which in our case is based on
the following result:
\begin{lemma}\sl
\label{FPDO}
 Let 
$\widetilde\varphi =\widetilde\varphi (x)$ be a real bounded
$C^\infty$ function on
$\R^n$ and let 
$p=p(x,\xi )\in S(1)$, that
extends to a bounded function, holomorphic with
respect to $\xi$ in a neighborhood of the set,
$$
\{ (x,\xi )\in 
{\rm Supp}\nabla\widetilde\varphi\times \C^n\; ;\; \vert \Im\xi\vert
\leq\vert\widetilde\nabla\varphi (x)\vert\}.
$$
Then, denoting by ${\rm Op}_h^L$ the left (or standard)
semiclassical quantization of symbols, the operator
$e^{\widetilde\varphi /h}{\rm Op}^L_h(p) e^{-\widetilde\varphi /h}$ is
uniformly bounded on
$L^2(\R^n)$, and for any $a\in C_0^\infty
(\R^n)$ and $N\geq 1$, one has,
\begin{eqnarray}
\label{pseudoformel}
&& \left( e^{\widetilde\varphi /h}{\rm
Op}^L_h(p)e^{-\widetilde\varphi /h}a\right) (x;h)\\
&& =
\sum_{\vert\alpha\vert\leq N}
\frac1{\alpha
!}\left(\frac{h}{i}\right)^{\vert\alpha\vert}
\partial_\xi^\alpha p(x,i\nabla \widetilde\varphi (x))
\partial_y^\alpha\left(a(y)e^{\Phi (x,y)/h}
\right)_{y=x} +{\mathcal O}(h^{N/2}),\nonumber
\end{eqnarray}
locally uniformly with respect to $x$, and
uniformly with respect to $h$ small enough.
Here, $\Phi (x,y):=\widetilde\varphi
(x)-\widetilde\varphi (y)-(x-y)\widetilde\nabla \varphi (x)$.
\end{lemma}
The proof of this lemma is rather standard
and we omit it (see
 e.g. \cite{Ma1}).
 
 Then, the construction can be performed by using the formal series given in (\ref{pseudoformel}) in order to define the formal action of $R(x,hD_x)$ on $w_1$. Afterwards, these constructions can be continued along the
integral curves of the  vector field
$\nabla_\xi
p_2(x,i\nabla\varphi (x))D_x=
2\nabla\varphi (x).\nabla_x$ (that is, along the minimal geodesic of $d$ starting at $0$),
as long as
$p_1 (x,i\nabla 
\varphi (x))$ does not vanish (that is, as long as these minimal geodesics stay inside the cirque $\Omega$). In that way, after resummation and multiplication by a cut-off function, we obtain a function $w_1$  of the form
(\ref{BKWinit}), that satisfies,
\be
\label{eqcirque}
Pw_1-\rho_1 w_1={\mathcal O}(h^\infty e^{-\varphi /h}),
\ee
locally uniformly in $\bigcup \gamma$, where the union is taken over all the minimal $d$-geodesics $\gamma$ coming from the well 0 and staying in $\Omega$. In particular, (\ref{eqcirque}) is satisfied in a neighborhood ${\mathcal N}_1$ of $\gamma^{(1)}$. 

\subsection{At the boundary of the cirque}
Now, we study the situation  near the point  $x^{(1)}\in \partial\Omega$. By  \cite{Pe}, Theorem 2.14, we know that there exists a neighborhood ${\mathcal V}_1$ of $x^{(1)}$ and two positive functions $\varphi_1, \varphi_2\in C^\infty ( {\mathcal V}_1)$ such that,
\begin{eqnarray*}
&&\varphi_1 = \varphi \mbox{ on } {\mathcal V}_1\cap\{ V_1 <V_2\};\\
&&\varphi_2= \varphi \mbox{ on } {\mathcal V}_1\cap\{ V_2 <V_1\};\\
&&|\nabla\varphi_j(x)|^2 = V_j(x) \quad (j=1,2);\\
&&\varphi_1 =\varphi_2 \mbox{ and } \nabla\varphi_1 =\nabla\varphi_2 \mbox{ on } {\mathcal V}_1\cap\partial\Omega;\\
&&\varphi_2(x) -\varphi_1(x) \sim d(x, \partial\Omega)^2.
\end{eqnarray*}
Actually, $\varphi_2$ is nothing but $d_2(0,x)$, where $d_2$ is the Agmon distance associated with the metric $V_2(x)dx^2$, and $\varphi_1$ is the phase function of the Lagrangian manifold obtained as the flux-out of $\{(x,\nabla\varphi_2(x))\, ;\, x\in {\mathcal V}_1\cap\partial\Omega \}$ under the Hamilton flow of $q_1(x,\xi):=\xi^2 - V_1(x)$.
\vskip 0.3cm
Then, we set,
\be
\psi := \frac12\left(\varphi_1+
\varphi_2\right) ,
\ee
and we consider 
the smooth function $z(x)$ defined for $x\in {\mathcal V}_1$ by,
\begin{eqnarray}
\label{z}
&& z(x)^2=2\left(\varphi_2(x)-
\varphi_1(x)\right) \nonumber\\
&& z(x)<0\mbox{ on } {\mathcal V}_1\cap\{ V_2 <V_1\}.
\end{eqnarray}

\vskip 0.2cm
In order to extend the WKB construction
(\ref{BKWinit}) accross $\partial\Omega$ near $x^{(1)}$, we still follow \cite{Pe} and try a formal ansatz of the form,
\be
\label{BKWcross}
w_2(x;h)=\sum_{k\geq 0}h^k\left( \alpha_k(x,h)
Y_{k,0}\left(\frac{z(x)}{\sqrt h}\right)
+{\sqrt h}\beta_k(x,h)Y_{k,1}
\left(\frac{z(x)}{\sqrt h}\right)\right)
e^{-\psi (x)/h},
\ee
where,
\be
\alpha_k(x,h)=\left(\begin{array}{clcr} 
h\alpha_{k,1}(x,h) \\ 
\alpha_{k,2}(x,h)\end{array}\right)
 \quad ;\quad
\beta_k(x,h)=\left(\begin{array}{clcr} 
\beta_{k,1}(x,h) \\ 
h\beta_{k,2}(x,h)\end{array}\right),
\ee
$\alpha_{k,j}$ and $\beta_{k,j}$ are formal
symbols of the form,
\be
\label{symbcross}
\sum_{l\geq 0}\sum_{m=0}^l h^l({\rm ln}h)^m
\gamma^{l,m}(x)
\ee
(with $\gamma^{l,m}$ smooth in ${\mathbf\omega}_1$), and for any $k\geq 0$ and $\varepsilon \in\C$, the function $Y_{k,\varepsilon}$ is the so-called Weber function, defined by,
\be
Y_{k,\varepsilon}(z)=\partial_\varepsilon^k
Y_{0,\varepsilon}(z)
\ee
where $Y_{0,\varepsilon}$ is the unique 
entire function with respect to $\varepsilon$ and $z$,
solution of the Weber equation,
\be
Y_{0,\varepsilon}''+(\frac12 -\varepsilon
-\frac{z^2}4)Y_{0,\varepsilon}=0
\ee
such that,  for $\varepsilon >0$, one has,
\be
Y_{0,\varepsilon}(z)\sim
\frac{\sqrt {2\pi}}{\Gamma (\varepsilon )}
e^{z^2/4}z^{\varepsilon -1}
\qquad (z\rightarrow +\infty ).
\ee
As it is shown in \cite{Pe}, Theorem 4.3, a resummation of
(\ref{BKWcross}) is possible up to an error
of order ${\mathcal O}(h^\infty e^{-  \varphi
/h})$ . 
\vskip 0.3cm
Now, since $\varphi$ is not $C^\infty$ (but only $C^1$) near $x^{(1)}$, we need to
find some generalization of lemma \ref{FPDO}. For
technical reasons, in the rest of this section we prefer
to work with the {\it right} semiclassical quantization
of symbols, that we denote by ${\rm Op}^R_h$.
\vskip 0.3cm
For $\nu_0>0$ and $g\in C^\infty (\R^{2n}\; ;
\R_+)$,
we denote by
$S_{\nu_0}(g(x,\xi ))$ the set of (possibly $h$-dependent)
functions $p\in C^\infty (\R^{2n})$ that
extend to holomorphic functions with respect to
$\xi$  in the strip,
$$
{\mathcal A}_{\nu_0}:=\{(x,\xi )\in  \R^n\times\C^n\; ;\; 
\vert\Im\xi\vert < \nu_0\},
$$
and such that,
 for all $\alpha\in\N^{2n}$, one has,
\be
\partial^\alpha p(x,\xi )
={\mathcal O}(g(x,{\rm Re}\xi )),
\ee
uniformly with respect to $(x,\xi )\in
{\mathcal A}_{\nu_0}$ and $h>0$ small enough. We also denote by $S_0(g)$ the analogous  space of smooth symbols obtained by substituting $\R^{2n}$ to ${\mathcal A}_{\nu_0}$, and ``smooth'' to ''holomorphic''.

We first show,
\begin{lemma}\sl
\label{FPDO'}
Let $\nu_0>0$,  $m\in\R$, $p=p(x,\xi )\in
S_{\nu_0}(\la
\xi\ra^m)$, and let
$\phi=\phi (x)$ be a real bounded Lipschitz
function on
$\R^n$ such that
$$
\Vert\nabla \phi (x)\Vert_{L^\infty}< \nu_0;\\
$$
Let also $a=a(x\; ;h)\in C^\infty (\R^n )$
be such that, for all $\alpha\in\N^n$,
$$
(hD_x)^\alpha a(x\; ;h)
={\mathcal O}(e^{-\phi (x)/h}),
$$
uniformly with respect to $h$ small enough
and  $x\in\R^n$. Then,
$$\left( {\rm Op}^R_h(p)a\right) (x\; ;h)
={\mathcal O}(e^{-\phi (x)/h})$$
uniformly with respect to $h$ small enough
and  $x\in\R^n$.
\end{lemma}
\begin{proof} - We write,
\be
e^{\phi (x)/h}{\rm Op}^R_h(p)a (x\; ;h)
=\frac1{(2\pi h)^n}\int e^{i(x-y)\xi /h
+  \phi (x) /h}
p(y,\xi )a(y\; ;h)dyd\xi,
\ee
and, following \cite{Sj},
we make the change of contour
of integration in $\xi$,
\be
\R^n\ni\xi\mapsto \xi+i\nu_1\frac{x-y}{\vert
x-y\vert} ,
\ee
where $\Vert\nabla \phi (x)\Vert_{L^\infty} <\nu_1<\nu_0$. We obtain,
\begin{eqnarray}
\label{e1}
&& e^{\phi (x)/h}{\rm Op}^R_h(p)a (x\; ;h)\\
&&
=\frac1{(2\pi h)^n}\int e^{i(x-y)\xi /h
}
p\left( y, \xi+i\nu_1\frac{x-y}{\vert
x-y\vert} \right)\theta (x,y\; ;h)dyd\xi,\nonumber
\end{eqnarray}
with,
$$
\theta (x,y\; ;h)=a(y\; ;h)e^{\left( 
\phi (x) -\nu_1\vert x-y\vert\right)/h}
={\mathcal O}(e^{\phi (x) -\phi (y)-\nu_1\vert x-y\vert /h}).
$$
Therefore,
\be
\label{e2}
\theta (x,y\; ;h)=
{\mathcal O}(e^{-\delta\vert x-y\vert /h}),
\ee
with $\delta =
\nu_1-\Vert\nabla\phi\Vert_{L^\infty}>0$.

Then, in the case $m<-n$, the result follows
immediately from (\ref{e1})-(\ref{e2}) (and standard estimates on
oscillatory integrals). In
the general case, we just write,
\be
{\rm Op}^R_h(p)={\rm Op}^R_h(p)(2\nu_0 -
h^2\Delta_x)^{-k} (2\nu_0 - h^2\Delta_x)^k,
\ee
with $k$ integer large enough (e.g. $k=1+\vert
[m]\vert+n$), and, since
${\rm Op}^R_h(p)(2\nu_0 -
h^2\Delta_x)^{-k}$ is a
semiclassical pseudodifferential operators with
($h$-dependent) symbol in
$S_{\nu_0}(\la\xi\ra^{m-2k})$ $\subset
S_{\nu_0}(\la\xi\ra^{-n-1})$, the result follows by applying
the previous case with $a$ replaced by $(2\nu_0 -
h^2\Delta_x)^ka$.
\end{proof}
\vskip 0.5cm
Now, in view of defining a formal
pseudodifferential calculus acting on
expressions such as (\ref{BKWcross}),  for $j=1,...,n$ and $x\in{\mathbf\omega}_1$, we set,
\be
A_j(x):=
\left(\begin{array}{clcr} 
\frac{\partial\varphi_2(x)}{\partial
x_j} & 0
\\  0 &
\frac{\partial\varphi_1(x)}{\partial
x_j}
\end{array}\right) \in {\mathcal M}_2(\R ) .
\ee
Then, for any $k\geq 0$, we have (see \cite{Pe}
formula (4.18)),
\begin{eqnarray}
\label{Weber}
&& \left( hD_{x_j}-iA_j(x)\right)
\left(\begin{array}{clcr} 
Y_{k,0}\left(\frac{z(x)}{\sqrt h}\right)
\\
Y_{k,1}
\left(\frac{z(x)}{\sqrt h}\right)
\end{array}\right)
e^{-\psi (x)/h}\\
&&\hskip 3cm  =\frac{\sqrt h}i\left( \partial_{x_j}z(x)\right)
\left(\begin{array}{clcr} 
kY_{k-1,1}\left(\frac{z(x)}{\sqrt h}\right)
\\
Y_{k,0}
\left(\frac{z(x)}{\sqrt h}\right)
\end{array}\right)
e^{-\psi (x)/h} \nonumber.
\end{eqnarray}
If $a$ and $b$ are (scalar) formal
symbols of the type (\ref{symbcross}), and $k\in\N$, we
set,
\be
I_k(a,b)(x\; ;\; h)= a(x\; ;\;  h)Y_{k,0}\left(
\frac{z(x)}{\sqrt h}\right) +b(x\; ;\;  h)Y_{k,1}\left(
\frac{z(x)}{\sqrt h}\right),
\ee 
and we make act any diagonal matrix-valued function
$B(x)={\rm diag}(B_1(x),B_2(x))$ $\in {\mathcal M}_2(\R )$ 
on
$I_k(a,b)(x\; ;h)$ by setting,
\be
 B(x)I_k(a,b)(x\; ;h):=I_k(B_1a,B_2b)(x\; ;h).
\ee
(this is possible since the $Y_{k,0}$ and $Y_{k,1}$ are
linearly independent). Then, using (\ref{Weber}), 
for all $j=1,...,n$ we have,
\begin{eqnarray}
\label{action}
&& \!\left( hD_{x_j}-iA_j(x)\right)\left(
I_k(a,b)e^{-\psi /h}\right)\\
&&\quad =\left(
I_k(hD_{x_j}a +{\sqrt {h}}bD_{x_j}z\; ,\;
hD_{x_j}b)
+I_{k-1}(0\; ,\; k{\sqrt {h}}aD_{x_j}z)
\right) e^{-\psi /h} .\nonumber
\end{eqnarray}
\par For  $A(x)=\left( A_1(x),...,A_n(x)\right)
\in\left( {\mathcal M}_2(\R )\right)^n$ and $\alpha \in \N^n$, we also use the notation,
\be
 A(x)^\alpha = A_1(x)^{\alpha_1}...
A_n(x)^{\alpha_n}
\in {\mathcal M}_2(\R ),
\ee
and we identify any $\xi\in\R^n$ with
$(\xi_1{\mathcal I}_2,...,\xi_n{\mathcal I}_2
)\in
\left( {\mathcal M}_2(\R )\right)^n$.

Then, we have,
\begin{lemma}\sl
\label{Reste}
Let $\nu_0> \sup_{x\in{\mathbf\omega}_1}\min (\sqrt{V_1(x)}, \sqrt{V_2(x)})$ and $m\in\R$. Then, for any $B=B(x,\xi )={\rm
diag}(B_1(x,\xi ),B_2(x,\xi ))\in
{\mathcal M}_2\left( S_{\nu_0}(\la
\xi\ra^m)\right)$,
$k\geq 0$, $a$ and $b$ in
$C_0^\infty ({\mathbf\omega}_1 )$, and
$\alpha\in\N^n$, one has,
\be
{\rm Op}^R_h\left( B(x,\xi )\left( \xi - iA (x)
\right)^\alpha\right) 
\left( I_k(a,b)
 e^{-\psi (x)/h}\right)
={\mathcal O}\left( \vert{\rm ln}h\vert^k
h^{\vert\alpha\vert /2} e^{- \varphi (x)/h}\right),
\ee
where the estimates holds uniformly for $h$ small enough and $x\in\R^n$.
\end{lemma}
\begin{proof} We prove it by induction on
$\vert\alpha\vert$, being careful  to the fact that we
do not have at disposal a symbolic calculus similar to
that of the usual pseudodifferential operators
(because of the special kind of action of the diagonal
matrices on $I_k(a,b)$, which actually does not commute
with the oscillatory integrations). We first notice that,
by \cite{Pe}-lemma 4.6, for
$\beta\in\N^n$ and
$j\in\{0,1\}$, one has,
\be
(hD_x)^\beta\left( 
Y_{k,j}\left(\frac{z(x)}{\sqrt h}\right)
e^{-\psi (x)/h}\right)
={\mathcal O}\left( \vert{\rm ln}h\vert^ke^{-
\varphi(x)/h}\right) .
\ee
As a consequence, the result for $\alpha = 0$
follows directly from Lemma \ref{FPDO'}. 

Now, assume it is true
for
$\vert\alpha\vert
\leq N$ ($N\in \N$ fixed arbitrarily), and let
$\gamma\in\N^n$,
$\vert\gamma\vert =1$. Then, for
$\vert\alpha\vert \leq N$, we have,
\begin{eqnarray}
&& {\rm Op}^R_h\left( B(x,\xi )\left( \xi - iA (x)
\right)^{\alpha +\gamma}\right)I_k(a,b)e^{-\psi (x)
/h}\\ 
&& = \frac1{(2\pi h)^n}\int e^{i(x-y)\xi
/h}F_\alpha (y,\xi )\left( \xi - iA(y)
\right)^\gamma I_k(a,b)(y)e^{-\psi (y)/h}dyd\xi,\nonumber 
\end{eqnarray}
with $F_\alpha (y,\xi )=B(y,\xi )\left( \xi - iA(y)
\right)^\alpha$. Now, assuming without loss of
generality that $\gamma =(1,0,...,0)$, and using the
fact that,
$$
\xi_1e^{i(x-y)\xi
/h}=-hD_{y_1}(e^{i(x-y)\xi
/h}),
$$
 we obtain (denoting by $F_{\alpha ,1}$ and
$F_{\alpha ,2}$ the two diagonal coefficients of
$F_\alpha$),
\vskip 0.5cm\noindent
$ {\rm Op}^R_h\left( B(x,\xi )\left( \xi - iA (x)
\right)^{\alpha +\gamma}\right)I_k(a,b)e^{-\psi (x)
/h}$
\begin{eqnarray*}= \frac1{(2\pi h)^n}\int e^{i(x-y)\xi
/h}(hD_{y_1}-iA_1(y))I_k(F_{\alpha ,1}(y,\xi )a(y),
F_{\alpha ,2}(y,\xi )b(y))\\
\times e^{-\psi (y)/h}dyd\xi,
\end{eqnarray*}
and therefore, by (\ref{action}),
\vskip 0.5cm\noindent
${\rm Op}^R_h\left( B(x,\xi )\left( \xi - iA (x)
\right)^{\alpha +\gamma}\right)I_k(a,b)e^{-\psi (x)
/h}$
\begin{eqnarray*}
&& = \frac1{(2\pi h)^n}\int e^{i(x-y)\xi
/h}\left[ I_k(hD_{y_1}F_{\alpha
,1}a+{\sqrt h}F_{\alpha ,2}bD_{y_1}z,hD_{y_1}F_{\alpha
,2}b)\right. \\
&& \left. \hskip 4.5cm +I_{k-1}(0,k{\sqrt h}F_{\alpha ,1}
aD_{y_1}z)\right]e^{-\psi (y)/h} dyd\xi\\
&& = h{\rm Op}^R_h(F_\alpha
)I_k(D_{y_1}a,D_{y_1}b)e^{-\psi /h} +h{\rm
Op}^R_h(D_{y_1}F_\alpha )I_k(a,b)e^{-\psi  /h}\\ 
&&\hskip
3.5cm +{\sqrt h}{\rm Op}^R_h({\rm diag}(F_{\alpha
,2},0))I_k(bD_{y_1}z,0)e^{-\psi  /h}\\
&&\hskip
3.5cm  +{\sqrt h}{\rm Op}^R_h({\rm
diag}(0,F_{\alpha ,1}))I_{k-1}(0,kaD_{y_1}z)e^{-\psi /h}.
\end{eqnarray*}
Then, applying the induction hypothesis (and using
the fact that $D_{y_1}F_\alpha$ is a sum of terms of the
type $B'(y,\xi)(\eta -iA(y))^\beta$ with
$\vert\beta\vert =\vert\alpha\vert -1$) this gives,
\begin{eqnarray*}
&&{\rm Op}^R_h\left( B(x,\xi )\left( \xi - iA (x)
\right)^{\alpha +\gamma}\right)I_k(a,b)e^{-\psi (x)
/h}\\
&& ={\mathcal O}\left(\vert {\rm
ln}h\vert^kh^{1+\frac{\vert\alpha\vert}2}+\vert {\rm
ln}h\vert^kh^{1+\frac{\vert\alpha\vert
-1}2} +\vert {\rm
ln}h\vert^kh^{\frac{1+\vert\alpha\vert}2}+\vert {\rm
ln}h\vert^{k-1}h^{\frac{1+\vert\alpha\vert}2}\right)
e^{-\varphi /h}\\
&& ={\mathcal O}\left(\vert {\rm
ln}h\vert^kh^{\frac{1+\vert\alpha\vert}2}\right) e^{-\varphi
/h}
\end{eqnarray*}
and the proof is complete.
\end{proof}
\vskip 0.5cm

Now, for any smooth function $f$ on $\R^n$, we set,
\be
f(A (x))={\rm diag}\left(
f(\frac{\partial\varphi_2(x)}{\partial x}),
f(\frac{\partial\varphi_1(x)}{\partial x})\right)
\in C^\infty (\R^n ,{\mathcal M}_2(\C )).
\ee
Then, for any $p\in
S_{\nu_0}(\la\xi\ra^m)$ and  
for any
$N\geq 1$, Taylor's formula gives,
\begin{eqnarray}
p(x,\xi ){\bf I}_2=\sum_{\vert\alpha\vert\leq N}
\frac1{\alpha !}
\partial_\xi^\alpha p(x, iA(x))
(\xi -iA (x))^\alpha \nonumber\\
+\sum_{\vert\alpha\vert
=N+1}B_\alpha (x,\xi )
(\xi -iA(x))^\alpha ,
\end{eqnarray}
where ${\bf I}_2$ is the 2$\times$2 identity matrix, and
the $B_\alpha$'s are in
${\mathcal M}_2\left( S_{\nu_0}(\la
\xi\ra^{m})\right)$.
In particular, using Lemma \ref{Reste},
 for any $a$ and $b$ in $C_0^\infty (
{\mathbf\omega}_1 )$, we obtain,
\vskip 0.5cm\noindent
${\rm Op}_h^R\left( p\right) 
\left( I_k(a,b)
 e^{-\psi /h}\right)$
\begin{eqnarray*}
&=&\sum_{\vert\alpha\vert\leq N}
\frac1{\alpha !}
{\rm Op}_h^R\left(\partial_\xi^\alpha p(x, iA(x))
(\xi -iA(x))^\alpha\right)
\left( I_k(a,b)
 e^{-\psi(x)/h}\right)\\
&& \hskip 8cm +{\mathcal O}(h^{N/2}e^{-\varphi/h})\\
&=& \sum_{{\vert\alpha\vert\leq N}\atop
{\beta\leq\alpha}}\frac1{i^{\vert\beta\vert}
\beta !(\alpha
-\beta )!}
{\rm Op}_h^R\left(\partial_\xi^\alpha p(x,iA(x))
A(x)^{\beta}\xi^{\alpha -\beta}\right)
\left( I_k(a,b)
 e^{-\psi(x)/h}\right)\\
&& \hskip 8cm +{\mathcal O}(h^{N/2}e^{-\varphi/h}),
\end{eqnarray*}
and thus, writing down as before the corresponding
oscillatory integral, in the same way we deduce,
\vskip 0.5cm\noindent
${\rm Op}_h^R\left( p\right) 
\left( I_k(a,b)
 e^{-\psi/h}\right)$
\begin{eqnarray}
\label{F1}
&=&\sum_{{\vert\alpha\vert\leq N}\atop
{\beta\leq\alpha}}\frac1{i^{\vert\beta\vert}
\beta !(\alpha
-\beta )!}(hD_x)^{\alpha
-\beta}\left[A(x)^{\beta}\partial_\xi^\alpha p(x,iA(x)) I_k(a,b)
 e^{-\psi /h}\right]\nonumber\\
&& \hskip 7cm +{\mathcal O}(h^{N/2}e^{-\varphi/h}) .
\end{eqnarray}

Now, for $M\in\Z$ and $\Omega\subset\R^n$ open,
we consider the space of sequences of formal
symbols,
\begin{eqnarray*}
S^M({\mathbf\omega}_1 ) := \{ a = (a_k)_{k\in\N}\; ;
a_k (x,h)=\sum_{l=-M}^\infty\sum_{m=0}^l
h^{l}({\rm ln}h)^m\gamma_k^{l,m}(x)\; ;\\
\gamma_k^{l,m}\in C^\infty ({\mathbf\omega}_1 )\} .
\end{eqnarray*}
and, for $a,b\in S^M({\mathbf\omega}_1 )$, we set,
\be
I(a,b)e^{-\psi /h}:=
\sum_{k\geq 0}h^k I_k(a_k,{\sqrt{h}}b_k)
e^{-\psi  /h} .
\ee
Using (\ref{action}), we see that, for $j=1,...,n$, the
action of $\left( hD_{x_j}-iA_j(x)\right)$ on such
formal series satisfies,
\be
\left( hD_{x_j}-iA_j(x)\right) I(a,b)e^{-\psi /h}
=I(L_j(a,b))e^{-\psi /h},
\ee
where $L_j$ is the operator,
\begin{eqnarray}
L_j\; :  S^M\times S^M &&\!\!\! \rightarrow
\quad S^{M-1}\times S^{M-1} \nonumber\\
(a,b) &&\!\!\! \mapsto \quad (\widetilde a^j,\widetilde
b^j)
\nonumber
\end{eqnarray}
defined by,
\begin{eqnarray}
&&
\widetilde a_k^j= hD_{x_j}a_k + hb_kD_{x_j}z\; ;
\nonumber\\
&&
\widetilde b_k^j = hD_{x_j}b_k+ (k+1)ha_{k+1}
D_{x_j}z ,
\end{eqnarray}
($k\in\N$).
In particular, using the notations $L=(L_1,...,L_n)$ and $L^\alpha =L_1^{\alpha_1}...L_n^{\alpha_n}$, for all $\alpha\in\N^n$, we have,
\be
\label{F2}
L^\alpha\;\;{\rm maps\;\;} S^M({\mathbf\omega}_1 )\times S^M({\mathbf\omega}_1)\;\;
{\rm into\;\;} S^{M-\vert\alpha\vert}({\mathbf\omega}_1
)\times S^{M-\vert\alpha\vert}({\mathbf\omega}_1 ) .
\ee
We also make naturally act  any smooth diagonal ${\mathcal
M}_2(\C )$-valued function $B(x)={\rm
diag}(B_1(x),B_2(x))$ on
$S^M\times S^M$ by setting,
\be
B(a,b)=(B_1a,B_2b),
\ee
and we define the formal action of a pseudodifferential
operator with symbol $p\in S_{\nu_0}(\la \xi\ra^m)$, on
expressions of the type $I(a,b)e^{-\psi
/h}$, by the formula,
\begin{eqnarray}
&&\hskip 0.8cm{\rm Op}_h^F (p)\left( I(a,b)e^{-\psi
/h}\right)\\
\label{FPDOcross}
&&=\sum_{{\alpha\in\N^n}\atop
{\beta\leq\alpha}}\frac1{i^{\vert\beta\vert}
\beta !(\alpha
-\beta )!}I\left( (iA(x)+L)^{\alpha 
-\beta}A(x)^{\beta}\partial_\xi^\alpha p(x, iA(x))
(a,b)\right)
 e^{-\psi /h}. \nonumber
\end{eqnarray}
Then, in view of Lemma \ref{Reste} and (\ref{F1}),
 we immediately obtain,
\begin{proposition}\sl
\label{resum}
Let $a,b\in S^M({\mathbf\omega}_1 )$ and denote by $\widetilde I(a,b)e^{-\psi  /h}$
any resummation of $I(a,b)e^{-\psi /h}$ up to
a ${\mathcal O}(h^\infty e^{-\varphi /h})$-error term. Then, for
any $\chi\in C_0^\infty ({\mathbf\omega}_1)$, the
quantity
${\rm Op}_h^R (p)\left(
\chi\widetilde I(a,b)e^{-\psi  /h}\right)$ is a resummation of
$\; {\rm Op}_h^F (p)\left( I(\chi a,\chi b)e^{-\psi
/h}\right)$, up to
a ${\mathcal O}(h^\infty e^{-\varphi /h})$-error term.
\end{proposition}
In particular, the operator $P$  naturally acts (up to ${\mathcal O}(h^\infty e^{-\varphi /h})$-error terms)
on expressions of the type,
\be
w_2=\left(\begin{array}{clcr} 
I(h\alpha_1,\beta_1) \\ 
I(\alpha_2,h\beta_2)\end{array}\right)
e^{-\psi /h},
\ee
where $\alpha_j=(\alpha_{j,k})_{k\geq 0}$
and $\beta_j=(\beta_{j,k})_{k\geq 0}$ are in
$S^0({\mathbf\omega}_1 )\;$ ($j=1,2$). 
\vskip 0.3cm
Writing down the equation
$\widetilde Pw_2=\rho_1 w_2$,  setting, 
$$
\alpha_{j,k} = \sum_{l\geq 0}\sum_{m=0}^l
h^l({\rm ln}h)^m\alpha_{j,k}^{l,m}(x),
$$
and the
analog formula for $\beta_{j,k}$, and identifying the coefficients of
$h^l({\rm ln}h)^m$ for $0\leq m\leq l\leq 1$, we find (denoting by $
 p=\left(\begin{array}{clcr} 
 p_1 +hr_{1,1} && h r_{1,2} \\ 
h r_{2,1} &&  p_2 +hr_{2,2} \end{array}\right)
$ the right-symbol of $P$),
\begin{eqnarray}
&&
p_1(x,i\nabla\varphi_2)\alpha_{1,0}^{0,0} +
r_{1,2}(x,i\nabla\varphi_2)\alpha_{2,0}^{0,0}
+\left[\frac1i\nabla_\xi
p_1(x,i\nabla\varphi_1)(\nabla
z)\right. \nonumber\\
\label{eq1}&& \hskip 3.5cm \left. +\frac12\la ({\rm
Hess}_\xi
p_1)(x,i\nabla\varphi_1)\nabla z ,\nabla
(\varphi_2 -\varphi_1)\ra\right]
\beta_{1,0}^{0,0}=0;\\
&& \left[\partial_\xi
p_1(x,i\nabla\varphi_1)D_x
-i(\nabla_x.\nabla_\xi
p_1)(x,i\nabla\varphi_1)\right.\nonumber \\
\label{eq2}
&&\hskip 6cm\left. +
r_{1,1}(x,i\nabla\varphi_1)-\rho_1\right]
\beta_{1,0}^{0,0}=0;\\
&&
p_2(x,i\nabla\varphi_1)\beta_{2,0}^{0,0} +
r_{2,1}(x,i\nabla\varphi_1)\beta_{1,0}^{0,0}
+\left[ \frac1i\partial_\xi
p_2(x,i\nabla\varphi_2)(\nabla
z)\right. \nonumber\\
\label{eq3}&&\hskip 3.5cm \left. +\frac12\la ({\rm
Hess}_\xi
p_2)(x,i\nabla\varphi_2)\nabla z ,\nabla
(\varphi_1
-\varphi_2)\ra\right]\alpha_{2,1}^{0,0}=0;\\
&& \left(\partial_\xi
p_2(x,i\nabla\varphi_2)D_x
-i(\nabla_x.\nabla_\xi
p_2)(x,i\nabla\varphi_2)\right.\nonumber\\
\label{eq4}
&&\hskip 6cm \left.
+r_{2,2}(x,i\nabla\varphi_2)-\rho_1\right)
\alpha_{2,0}^{0,0}=0;
\end{eqnarray}
(Here we also have used the fact  that $\rho\sim
\sum_{k\geq 1}h^k\rho_k$ as $h\rightarrow 0$.)
Identifying the other coefficients, one
obtains a series of equations that (in a way
similar to \cite{Pe}-section 4) can be solved in
 ${\mathcal V}_1$ (possibly after having shrunk  it a little bit around $x^{(1)}$), and in such a way
that one also has,
\be
\widetilde w_2-\widetilde w_1={\mathcal O}(h^\infty
e^{-\varphi /h})
\quad {\rm locally\; uniformly\; in}\quad {\mathcal V}_1\cap \{V_2 <V_1\}. 
\ee
where $w_1$ is defined in (\ref{BKWinit}), and
$\widetilde w_1$, $\widetilde w_2$ are resummations of
$w_1$ and $w_2$. Among other things, this implies,
\be
\label{eq5}
\alpha_{2,0}^{0,0}=a_{2,0}
\quad {\rm in}\quad {\mathcal V}_1\cap \{V_2 <V_1\} .
\ee
Moreover, we see on
(\ref{eq2}) and (\ref{eq4}) that
$\beta_{1,0}^{0,0}$ (respectively
$\alpha_{2,0}^{0,0}$) are solutions of a differential
equation of order 1 on each integral curve of
the real vector field
$\nabla\varphi_1(y).\nabla_y$ (respectively
$\nabla\varphi_2(y).\nabla_y$).
In particular, because of the ellipticity
of $a_{2,0}$, we deduce from (\ref{eq4}) and
(\ref{eq5}) that we have,
\be
\label{ellip}
\alpha_{2,0}^{0,0} \quad {\rm never\;
vanishes\; in \;} \; {\mathcal V}_1.
\ee 
Now, Assumption 6 implies that, if $\gamma\in G_0$, then,
\be
 r_{1,2}(x,i\nabla\varphi_2)\not= 0\quad \mbox{ on } {\mathcal V}_1.
 \ee 
Since 
$p_1(y,i\nabla\varphi_2)=
p_1(y,i\nabla\varphi_1)=0$ on ${\mathbf\omega}_1\cap \partial\Omega$, we
deduce from (\ref{eq1}) and (\ref{ellip}) that, if $\gamma\in G_0$, then
$\beta_{1,0}^{0,0}$ does not vanish on ${\mathbf\omega}_1\cap \partial\Omega$. As before, because of
(\ref{eq2}) (and the fact that $R(x,hD_x)$ is formally selfadjoint), this implies,
\be
\label{ellip'}
\mbox{ If } \gamma\in G_0, \mbox{ then, } \beta_{1,0}^{0,0} \quad {\rm never\;
vanishes\; in \;}{\mathcal V}_1 .
\ee 

\subsection{In the island, outside the cirque}
Now, we look at what happens on $\gamma^{(2)}$, and, at first, near $x^{(1)}$.
Using the asymptotics of $Y_{k,\varepsilon}
(z/{\sqrt h})$ given in \cite{Pe}-section 4, one
also finds that, in ${\mathcal V}_1\cap\{V_1<V_2\}$, $w_2$
can be formally identified with,
\be
w_3(x,h)={\sqrt {2\pi
h}}\left(\begin{array}{clcr}  b_1(x,h) \\ 
hb_2(x,h)\end{array}\right) 
e^{-\varphi
(x)/h}
\ee
where $b_1,b_2$ are symbols of the form,
\be
\label{symblog}
b_j(x;h)=\sum_{l\geq 0}
\sum_{m=0}^lh^l({\rm ln}h)^mb_j^{l,m}(x)
\ee
($j=1,2$), with
$b_j^{l,m}\in
C^\infty ({\mathcal V}_1\cap\{V_1<V_2\})$, in the sense that, for any
resummations
$\widetilde w_2$ and $\widetilde w_3$ of $w_2$
and $w_3$, one has,
\be
\widetilde w_2-\widetilde w_3={\mathcal O}(h^\infty 
e^{-\varphi
/h})
\quad {\rm locally\; uniformly\; in}\quad \Omega\cap \Gamma_+ . 
\ee
Moreover, one also has,
\be
b_1^{0,0}=\beta_{1,0}^{0,0}
\ee
which, by (\ref{ellip'}), shows that, when $\gamma\in G_0$, $b_1$ is
elliptic in ${\mathcal V}_1\cap\{V_1<V_2\}$.
\vskip 0.3cm
Since $p_2(x,i\nabla\varphi
(x))\not= 0$  in $\{V_1<V_2\}$, we can
formally solve the equation $
Pw_3 =\rho_1 w_3$, and we see again that
$b_1$ and
$b_2$ can be continued along the integral
curves of $\nabla\varphi$, as long as these curves stay inside $\{V_1<V_2\}$ and $\varphi_1$
does not develop caustics. In particular, they can be continued in a neighborhood ${\mathcal N}_2$ of $\gamma^{(2)}$, and
 the continuation of $b_1$ remains
elliptic in $\Omega_2$.

\vskip 0.4cm
Clearly, the previous steps can be repeated  near $x^{(2)}$, $\gamma^{(3)}$, etc... (in the case $N_\gamma\geq 3$), up to finally reach $\gamma^{(N_\gamma +1)}$, obtaining in that way (after having pasted everything in a standard way by using a partition of unity) a function ${\mathbf w}(x,h)$, smooth on a neighborhood ${\mathcal N}(\gamma)$ of $\gamma$ in $\ddot O$, and satisfying,
$$
(P-\rho_1){\mathbf w} ={\mathcal O}(h^\infty e^{-\varphi /h}),
$$
locally uniformly in ${\mathcal N}(\gamma)$. Moreover, ${\mathcal N}(\gamma)$ can be decomposed into,
$$
{\mathcal N}(\gamma) = {\mathcal N}_1\cup {\mathcal V}_1\cup\dots \cup {\mathcal V}_{N_\gamma}\cup {\mathcal N}_{N_\gamma +1},
$$
where,
for all $j$, ${\mathcal V}_j$ is  a neighborhood  of $x^{(j)}$ and ${\mathcal N}_j$ is a neighborhood  of $\gamma^{(j)}$, in such a way that, in each ${\mathcal N}_j$, $\mathbf w$ admits a WKB asymptotics of the form,
\be
{\mathbf w}(x;h)\sim h^{\frac{j-1}2}\left(\begin{array}{clcr} 
h^{\frac{1-(-1)^j}2}a_1^{(j)}(x,h) \\ 
h^{\frac{1+(-1)^j}2}a_2^{(j)}(x,h)\end{array}\right) e^{-\varphi
(x)/h},
\ee
where $a_1^{(j)}$ and $a_2^{(j)}$ are symbols of the same form as in (\ref{symblog}), and $a_1^{(j)}$ is elliptic if $j$ is even, while $a_2^{(j)}$ is elliptic if $j$ is odd (in particular, $a_1^{(N_\gamma + 1)}$ is elliptic). On the other hand, in each ${\mathcal V}_j$, $\mathbf w$ can be representated by means of the Weber function, in a way similar to that of (\ref{BKWcross}).

\subsection{At and after the boundary of the island}
Let us denote by
$$
x_\gamma\in \gamma \cap \partial\ddot O,
$$
the point of type 1 where $\gamma$ touch the boundary of the island.
When $x\in\gamma\cap \ddot O$ is close enough to $x_\gamma$, we know from the previous subsection that the asymptotic solution $\mathbf w$ is of the form,
\be
{\mathbf w}(x;h)\sim h^{\frac{N_\gamma}2}\left(\begin{array}{clcr} 
b_1(x,h) \\ 
h b_2(x,h)\end{array}\right) e^{-\varphi
(x)/h},
\ee
where $b_1,b_2$ are smooth symbols on ${\mathcal N}_{N_\gamma +1}$, of the same form as in (\ref{symblog}), and $b_1$ is elliptic. Moreover, as $x$ approaches $x_\gamma$, $b_1$ and $b_2$ (together with $\varphi$) develop singularities on some set ${\mathcal C}$ (called the caustic set). However, following an idea of \cite{HeSj2}, we can represent
$h^{-\frac{N_\gamma}2}e^{S/h}w$ in the integral (Airy) form,
\begin{equation}
\label{airy}
I[c_1,c_2](x,h)=h^{-1/2}\int_{\gamma(x)}\left(\begin{array}{clcr} 
c_1(x',\xi_n,h) \\ 
h c_2(x',\xi_n,h)\end{array}\right)e^{-(x_n\xi_n+
g(x',\xi_n))/h} d\xi_n,
\end{equation}
where we have used local Euclidean coordinates $(x',x_n)\in \R^{n-1}\times \R$ centered at $\gamma\cap\partial\ddot O$, such that $V_1(x) =-C_0x_n+{\mathcal O}(x^2)$ near this point.
For $x$ in $\ddot O$ close to $\gamma\cap\partial\ddot O$, the phase function
$\xi_n\mapsto x_n\xi_n+g(x',\xi_n)$ admits two real critical points that are close to $0$.
Then, choosing  conveniently the $x$-dependent interval $\gamma(x)$, the steepest descent method at one of these points gives us the asymptotic
expansion of
$I[c_1,c_2]$. Comparing this with the symbols $b_1$ and $b_2$, one can determine 
$c_1$ and $c_2$ so that the asymptotic expansion of
$h^{-\frac{N_\gamma}2}e^{S/h}w$ coincides with that of $I[c_1,c_2]$ in $\ddot O$. In particular, when $\gamma\in G_0$, one finds that $c_1$ remains elliptic near 0.
\vskip 0.3cm
At this point, since we did not assume any analyticity of the potentials near $\ddot O$, we have to follow the methods of \cite{FLM} where a similar situation is considered. Indeed, following the constructions of \cite{FLM}, Section 4 (that are made in the scalar case, but can be generalized without problem to our vectorial case), we see that there exists a constant $\delta >0$ such that, for any $N\geq 1$, one can construct a (vectorial) function $w_N$, smooth on the set,
\be
{\mathcal W}_N(\gamma) :=\{ |x-x_\gamma|<\varepsilon\}\cap \{ {\rm dist}(x,\ddot O)<2 (Nk)^{2/3}\}
\ee
 with $\varepsilon >0$ small enough (recall from (\ref{defk}) that $k=|h\ln h|$), such that 
 (see \cite{FLM}, Propositions 4.5 and 4.6),
\begin{itemize}
\item $(P-\rho_1)w_N ={\mathcal O}(h^{\delta N}e^{-\Re \widetilde \varphi_N /h})$ uniformly in ${\mathcal W}_N(\gamma) $;
\item For any $\alpha\in\Z_+^n$, there exists $m_\alpha\geq 0$ independent of $N$ such that,
$$
\partial_x^\alpha w_N ={\mathcal O}(h^{-m_\alpha}e^{-\Re \widetilde \varphi_N /h})
$$
 uniformly in ${\mathcal W}_N(\gamma) $;
 \item $w_N$ can be represented by an integral of the form (\ref{airy}) (with $\gamma (x)=\gamma_N(x)$ depending on $N$) in all of ${\mathcal W}_N(\gamma) $;
 \item  $w_N = {\mathbf w}$ in ${\mathcal N}_{N_\gamma +1}\cap {\mathcal W}_N(\gamma) $;
 \item For any large enough $L$, there exist $C_L>0$ and $\delta_L>0$, both independent of $N$ such that, uniformly in ${\mathcal W}_N(\gamma) \cap \{ {\rm dist} (x, \ddot O)\geq (Nk)^{2/3}\}$, one has,
\begin{eqnarray}
\label{summarywkb}
&& \hskip 0.8cm w_N(x,h)\\
\nonumber
&&=h^{\frac{N_\gamma}2}
\left (\sum_{{\ell=0}\atop{0\leq m\leq \ell}}^{L+[Nk/C_Lh]} h^\ell(\ln h)^m \left(\begin{array}{clcr} 
f_{1,N}^{\ell,m}(x) \\ 
h f_{2,N}^{\ell,m}(x) \end{array}\right)+{\mathcal O}(h^{\delta_L N}+h^{L})\right )e^{-\widetilde\varphi_N(x)/h},
\end{eqnarray}
as $h\to 0$,
with $f_{1,N}^{\ell,m}(x) , f_{2,N}^{\ell,m}(x) $ independent of $h$, and of the form,
\be
\label{taj}
\widetilde f_{j,N}^{\ell,m} (x) = ({\rm dist}(x,{\mathcal C}))^{-3\ell/2-1/4}\beta_{j,N}^{\ell,m}(x,{\rm dist}(x,{\mathcal C})),
\ee 
($j=1,2$) where $\beta_{j,N}^{\ell,m}$ is  smooth near $(x_\gamma, 0)$, and $\beta_1^{\ell,m}(x_\gamma, 0)\not=0$ in the case $\gamma\in G_0$. 
\end{itemize}
Here, $\widetilde\varphi_N$ is a (complex-valued) $C^1$ function on ${\mathcal W}_N(\gamma) $, smooth on ${\mathcal W}_N(\gamma) \backslash{\mathcal C}$, such that (see \cite{FLM}, Lemma 4.1),
\begin{itemize}
\item $\widetilde\varphi_N =\varphi + {\mathcal O}(h^\infty)$ uniformly  in ${\mathcal N}_{N_\gamma +1}\cap {\mathcal W}_N(\gamma) $;
\item $(\nabla \widetilde\varphi_N)^2 =V_1(x) + {\mathcal O}(h^\infty)$ uniformly  in ${\mathcal W}_N(\gamma) $;
\item  There exists $\varepsilon (h)={\mathcal O}(h^\infty )$ real, such that, for $x\in  {\mathcal W}_N(\gamma) \backslash \ddot O$, one has,
\begin{equation}
\label{reimphi}
\Re\widetilde\varphi_N (x)\geq S-\varepsilon (h);
\end{equation}
\item One has,
$$
\Im \nabla \varphi_N(x) = -\nu_N(x)\sqrt{{\rm dist}(x,{\mathcal C})}\hskip 5pt\nabla {\rm dist}(x,{\mathcal C})+{\mathcal O}({\rm dist}(x,{\mathcal C})),
$$
uniformly with respect to $h>0$ small enough and $x\in  {\mathcal W}_N(\gamma) \backslash \ddot O$,  with $\nu_N(x)\geq \delta$.
\end{itemize}
\vskip 0.3cm
The previous results show that we can extend ${\mathbf w}$ by taking $w_N$ in ${\mathcal W}_N(\gamma) $, and  we obtain in that way a function  ${\mathbf w}_N$ smooth on  ${\mathcal N}(\gamma)\cup {\mathcal W}_N(\gamma) $,
such that $(P-\rho_1){\mathbf w}_N ={\mathcal O}(h^{\delta N}e^{-\Re \widetilde \varphi /h})$ uniformly in ${\mathcal N}(\gamma)\cup {\mathcal W}_{N}(\gamma) $. Note that, thanks to Assumption 4, the number $N_\gamma$ is constant on each connected component of $\Gamma$.

\section{Agmon estimates}\label{secagm}

\subsection{Preliminaries} In order to perform Agmon estimates in the same spirit as in \cite{HeSj1}, we need some preliminary results because of the fact that we have to deal with pseudodifferential operators (and not only Schr\"odinger operators). For this reason, we prefer to work with $C^\infty$ weight functions (instead of Lipschitz ones), and the idea is to take $h$-dependent regularizations of Lipschitz weights.
\vskip 0.3cm
At first, we need,
\begin{proposition}\sl\label{poidssing}
\label{poidsvar}
Let  $\nu_0 >0$, $m\geq 0$, $a=a(x,\xi )\in S_{\nu_0}(\langle\xi\rangle^{2m})$.
For $h>0$ small enough, let also
$\Phi_h\in C^\infty (\R^n)$ real valued, such that,
\be
\label{estreg1}
 \sup\vert\nabla\Phi_h
\vert < \nu_0,
\ee
and, for 
 any multi-index $\alpha\in\N^n$ with $|\alpha|\geq 2$,
\be
\label{estreg}
\partial^\alpha\Phi_h (x)= {\mathcal O}\left( h^{1-\vert
\alpha\vert} \right),
\ee
uniformly for $x\in\R^n$ and $h>0$ small enough. 
Then,  for any $\widetilde\Sigma\subset \R^n$ with dist$(\Sigma ,\R^n\backslash \widetilde\Sigma )>0$,  the operator $e^{\Phi_h /h}Ae^{-\Phi_h /h}:=e^{\Phi_h /h} {\rm Op}_h^W(a)e^{-\Phi_h /h}$  satisfies,
\be
\label{estA}
\Vert e^{\Phi_h /h}Ae^{-\Phi_h /h}u\Vert_{L^2} \leq C_1  \Vert\langle hD_x\rangle^m u\Vert_{L^2}
\ee
uniformly
for all $h>0$ small enough and $u\in H^m(\R^n )$.
\end{proposition}
\begin{proof} For $u\in C_0^\infty (\R^n)$, we write,
$$
e^{\Phi /h}Ae^{-\Phi /h}u(x) =\frac1{(2\pi h)^n}\int e^{i(x-y)\xi /h + (\Phi (x)-\Phi (y))/h}
a(\frac{x+y}2, \xi )u(y)dyd\xi,
$$
and  the property (\ref{estreg1}) shows that we can make the change of contour of integration given by,
$$
\R^n\ni \xi \mapsto \xi +i\Psi (x,y),
$$
where $\Psi (x,y):= \int_0^1 \nabla\Phi ((1-t)x +ty)dt$ (in particular, one has: $\Phi (x)-\Phi (y) = (x-y)\Psi (x,y)$). Then, denoting by ${\rm Op}_h$  the semiclassical quantization of symbols depending on
$3n$ variables (see e.g. \cite{Ma2} Section 2.5), we obtain,
$$
e^{\Phi /h}Ae^{-\Phi /h}= {\rm Op}_h\left(a(\frac{x+y}2, \xi +i\Psi (x,y) )\right),
$$
and, using (\ref{estreg}), we see that, for any $\alpha, \beta, \gamma\in \Z_+^n$, we have
\be
\label{opconjug}
\partial_x^\alpha\partial_y^\beta \partial_\xi^\gamma\left(a(\frac{x+y}2, \xi +i\Psi (x,y) )\right) ={\mathcal O}(h^{-|\alpha+\beta|}\la\xi\ra^m).
\ee
Then, the results is an easy consequence of
 the Calder\'on-Vaillancourt Theorem: see, e.g., \cite{Ma2}, exercise 2.10.15.
\end{proof}
\vskip 0.2cm
We also need,
\begin{proposition}\sl\label{regulLip}
Let $\phi$ and $V$ be two  bounded real-valued Lipschitz functions on $\R^n$, such that $|\nabla\phi(x)|^2\leq V(x)$ almost everywhere. Let also $\chi_1\in C_0^\infty (\R^n; [0,1])$ supported in the ball $\{ |x|\leq 1\}$, such that $\int \chi_1 (x)dx =1$. For any $h>0$, we set $\chi_h(x)=h^{-n}\chi (x/h)$. Then, the smooth function,
$$
\phi_h := \chi_h\, \ast \phi
$$
(where $\ast$ stands for the standard convolution) satisfies,
\begin{itemize}
\item
$\phi_h = \phi +{\mathcal O}(h)$ uniformly for  $h>0$ small enough and $x\in\R^n$;\
\item For all $x\in\R^n$, one has
$|\nabla\phi_h(x)|^2 \leq V(x) + h\Vert \nabla V\Vert_{L^\infty}$; 
\item For all $\alpha\in \Z_+^n$ with $|\alpha|\geq 1$, one has 
$\partial^\alpha \phi_h ={\mathcal O}(h^{1-|\alpha|})$.
\end{itemize}
\end{proposition}
The proof of this proposition is very standard and almost obvious, and we leave it to the reader. Observe that, in particular, $\phi_h$ satisfies the estimates (\ref{estreg}).
\subsection{Agmon estimates}
As a corollary of the two previous propositions, we have,
\begin{corollary}
\label{AGMON}
Let $\phi$ and $\phi_h$ be as in Proposition \ref{regulLip}, with $V=\min(V_1,V_2)_+$. Then, for any $u=(u_1,u_2)\in H^2(\R^n)\oplus H^2(\R^n)$, one has,
\begin{eqnarray*}
\Re \la e^{\phi_h/h}Pu, e^{\phi_h/h}u\ra \geq \Vert h\nabla (e^{\phi_h/h}u)\Vert^2 +\sum_{j=1}^2\la (V_j -|\nabla\phi_h|^2)e^{\phi_h/h}u_j, e^{\phi_h/h}u_j\ra \\- C_Rh(\Vert e^{\phi_h/h}u\Vert^2+\Vert h\nabla (e^{\phi_h/h}u)\Vert^2),
\end{eqnarray*}
where $C_R>0$ is a constant that depends on $R(x,hD_x)$, $\chi_1$ and $\sup|\nabla\phi|$ only.
\end{corollary}
\begin{proof} It is standard (and elementary) to show that,
\begin{eqnarray*}
&&\Re \la e^{\phi_h/h}(-h^2\Delta + V_j)u_j, e^{\phi_h/h}u_j\ra \\
&& \hskip 3cm = \Vert h\nabla (e^{\phi_h/h}u_j)\Vert^2 +\la (V_j -|\nabla\phi_h|^2)e^{\phi_h/h}u_j, e^{\phi_h/h}u_j\ra.
\end{eqnarray*}
Therefore, it is enough to estimate $ \la e^{\phi_h/h}R(x,hD_x)u, e^{\phi_h/h}u\ra$. Applying Proposition \ref{poidssing} , we see that the operator $e^{\phi_h/h}R(x,hD_x)e^{-\phi_h/h}\la hD_x\ra^{-1}$ is uniformly bounded on $L^2$.

Moreover, since the constants appearing in the estimates (\ref{opconjug}) depend on $a$, $\alpha$, and on the estimates on  the $\partial^\beta\Phi$'s only, we see that the norm of $e^{\phi_h/h}R(x,hD_x)e^{-\phi_h/h}\la hD_x\ra^{-1}$ depends on $r$ and
on estimates on  $\partial^\beta (\chi_h\ast\nabla\phi) = (\partial^\beta\chi_h)\ast \nabla\phi$ ($\,|\beta|\leq |\alpha|$) only. Since the latter depend on  $\alpha$, $\chi_1$ and $\sup|\nabla\phi|$ only, the result follows.
\end{proof}
\section{Global asymptotic solution}\label{secglob}
The constructions of Section \ref{constWKB} can be done in a neighborhood of any minimal geodesic $\gamma\in G$, and give rise (after having pasted them together with a partition of unity)  to an asymptotic solution (still denoted by ${\mathbf w}_N$) on a neighborhood of $\bigcup_{\gamma\in G}\gamma$. Now, we plan to extend this solution to a whole ($h$-dependent) neighborhood of $\{V_1\geq 0\}$, by using a modified selfadjoint operator with discrete spectrum near 0. 
\vskip 0.3cm
At first, we fix $\varepsilon_0>0$ sufficiently small, and a cut-off function $\chi_0\in C_0^\infty (\ddot O; [0,1])$ such that,
$$
\chi_0(x) = 1  \mbox{ if }  V_1(x)\geq 2\varepsilon_0\,;\, \chi_0(x) = 0 \mbox{ if } V_1(x)\leq \varepsilon_0,
$$
and we set,
\be
\label{V1mod}
\widetilde V_1  := \chi_0 V_1 + \varepsilon_0(1-\chi_0).
\ee
In particular, $\widetilde V_1$ coincides with $V_1$  on $\{ V_1\geq 2\varepsilon_0\}$, and we have $\widetilde V_1\geq \varepsilon_0$ everywhere. Then, we define $\widetilde P_1:= -h^2\Delta + \widetilde V_1$, and we consider the selfadjoint operator,
\be
\widetilde P= 
\left(\begin{array}{cc}
\widetilde P_1 & 0\\
0 & P_2
\end{array}\right) + hR(x,hD_x).
\ee
By construction, for all $C>0$ and $h$ small enough, the spectrum of $\widetilde P$ is discrete in $[-Ch, Ch]$, and a straightforward adaptation of  the arguments used in \cite{HeSj1} shows that its first eigenvalue $E_1$ admits the same asymptotics as $\rho_1$ as $h\to 0_+$. We denote by ${\mathbf v}$ its first normalized eigenfunction, and by ${\mathcal N}_0\subset \{ V_1>2 \varepsilon_0\}$ some fix neighborhood of $\bigcup_{\gamma\in G}\cap \{ V_1>2 \varepsilon_0\}$ where the asymptotic solution ${\mathbf w}_N$ is well defined. We have,

\begin{proposition}\sl 
\label{patchvw}
There exists $\theta_0\in\R$ independent of $h$, such that,
for any compact subset $K$ of ${\mathcal N}_0$, and for any $\alpha\in\Z_+^n$, one has,
$$
\Vert e^{\varphi /h}\partial^\alpha(e^{i\theta_0}{\mathbf v}-h^{\frac{n}4}{\mathbf w}_N)\Vert_K  ={\mathcal O}(h^\infty).
$$
\end{proposition}
\begin{proof} The existence of $\theta_0$ such that $\partial^\alpha(e^{i\theta_0}{\mathbf v}-h^{\frac{n}4}{\mathbf w}_N)={\mathcal O}(h^\infty)$ uniformly near 0, is a consequence of \cite{HeSj1}, Proposition 2.5, and standard Sobolev estimates. Let $\chi\in C_0^\infty ( {\mathcal N}_0; [0,1])$, with $\chi =1$ in a neighborhood of  $K\cup\{0\}$. Following \cite{HeSj1, Pe}, we plan to apply Corollary \ref{AGMON} to $u:= \chi (e^{i\theta_0}{\mathbf v}-h^{\frac{n}4}{\mathbf w}_N)$, with a suitable weight function $\phi$. Let us first observe that, using Corollary \ref{AGMON},  for any $\varepsilon >0$, one has,
\be
\label{agmonfaible}
\Vert e^{(1-\varepsilon)\widetilde\varphi /h}\la hD_x\ra{\mathbf v}\Vert_{H^1} ={\mathcal O}(1).
\ee
where $\widetilde\varphi(x) \geq\varphi (x)$ is the Agmon distance associated with $\min(\widetilde V_1, V_2)$ between 0 and $x$.
Now, for $C\geq 1$ arbitrarily large, we define,
$$
\phi (x) := \min (\phi_1, \phi_2),
$$
where,
\begin{eqnarray*}
&&\phi_1(x):=\left\{\begin{array}{ll}
\varphi (x) - Ch\ln (\varphi (x)/h) & \mbox{ if } \varphi (x)\geq Ch;\\
\varphi (x) - Ch\ln C & \mbox{ if } \varphi (x)\leq Ch,
\end{array}\right.\\
&&\phi_2(x):=\left\{\begin{array}{ll}
\displaystyle{\inf_{\chi(y)\not= 1}}(1-2\varepsilon)(\varphi (y)+d(y,x)) &\mbox{ if } x\in \supp \chi;\\
(1-2\varepsilon)\varphi (x) &\mbox{ if } x\notin \supp\chi.
\end{array}\right.
\end{eqnarray*}
Here, $\varepsilon>0$ is taken sufficiently small in order to have $\phi_2(x) > \varphi(x)$ when $x\in K$. Then, $\phi$ is Lipschitz continuous,  and one has $\phi = \phi_1$ on $K$, and $\phi =\phi_2$ on $\R^n\backslash\{\chi =1\}$. Moreover, one sees as in \cite{Pe} (proof of Theorem 5.5) that, if we set $V:=\min(V_1,V_2)$, $\phi$ satisfies,
\begin{eqnarray*}
|\nabla \phi|^2 = V \mbox{ in } \{ \varphi \leq Ch\};\\
|\nabla \phi|^2\leq V - \delta_0Ch \mbox{ in } \{ \varphi \geq Ch\},
\end{eqnarray*}
where $\delta_0 =\inf_{x\in \supp\chi\,; \,x\not=0} (V(x)/\varphi(x) ) >0$. As a consequence, by Proposition \ref{regulLip}, the regularized $\phi_h$ of $\phi$ satisfies,
\begin{eqnarray*}
|\nabla \phi_h|^2 \leq V +h\Vert V\Vert_{L^\infty}\mbox{ in } \{ \varphi \leq Ch\};\\
|\nabla \phi_h|^2\leq V - (\delta_0C -\Vert V\Vert_{L^\infty})h \mbox{ in } \{ \varphi \geq Ch\}.
\end{eqnarray*}
Then, choosing $C$ sufficiently large, and setting $u:= \chi (e^{i\theta_0}{\mathbf v}-h^{\frac{n}4}{\mathbf w}_N)$,  we see that Corollary \ref{AGMON} implies,
 \be
\label{compDR}
\Vert h\nabla(e^{\phi_h/h}u)\Vert^2 + C' h\Vert e^{\phi_h/h}u\Vert_{\{\varphi\geq Ch\}}^2 \leq \la e^{\phi_h/h}(\widetilde P-E_1)u, e^{\phi_h/h} u\ra,
\ee
with $C'=C'(C)$ arbitrarily large. Moreover, if $\widetilde \chi\in C_0^\infty ({\mathcal N}_0)$ is such that $\widetilde\chi \chi = \chi$, we have,
$$
(\widetilde P-E_1)u=[\widetilde P, \chi]\widetilde\chi u+{\mathcal O}(h^\infty e^{-\varphi/h}),
$$
and since $\phi_h =(1-2\varepsilon )\varphi +{\mathcal O}(h)$ on $\supp\nabla\chi$, $\min_{\supp\nabla\chi}\varphi =:\delta_1>0$, and, by Proposition \ref{poidssing}, the operator $e^{\phi_h/h}[R,\chi]e^{-\phi_h/h}$ is uniformly bounded, we obtain (using also (\ref{agmonfaible})),
\begin{eqnarray*}
\la e^{\phi_h/h}(\widetilde P-E_1)u, e^{\phi_h/h}u\ra&=&{\mathcal O}\left(\Vert e^{(1-\varepsilon)\varphi /h }\la hD_x\ra\widetilde\chi u\Vert_{\supp\nabla\chi}^2 + h\Vert e^{\phi_h/h}u\Vert ^2\right)\\
&=& {\mathcal O}\left( e^{-2\varepsilon\delta_1/h}+ h\Vert e^{\phi_h/h}u\Vert ^2\right).
\end{eqnarray*}
Inserting this estimate into (\ref{compDR}), and taking $C$ sufficiently large,
this permits us to obtain,
$$
\Vert h\nabla(e^{\phi_h/h}u)\Vert^2 + h\Vert e^{\phi_h/h}u\Vert^2={\mathcal O}(e^{-2\varepsilon\delta_1/h}+ \Vert e^{\phi_h/h}u\Vert_{\{\varphi\leq Ch\}}^2).
$$
In particular, since $\phi_h = \phi_1 + {\mathcal O}(h)$ on $K$, and $\phi_h=(1-2\varepsilon)\varphi \leq Ch$ on $\{\varphi\leq Ch\}$,
$$
\Vert h^C\varphi^{-C}e^{\varphi/h}h\nabla u\Vert_K^2+\Vert h^C\varphi^{-C}e^{\varphi/h}u\Vert_K^2={\mathcal O}(e^{-2\varepsilon\delta_1/h} +  \Vert u\Vert_{\{\varphi\leq Ch\}}^2).
$$
Therefore,
$$
\Vert e^{\varphi/h}\nabla u\Vert_K^2+\Vert e^{\varphi/h}u\Vert_K^2={\mathcal O}(h^\infty),
$$
and the result follows by standard Sobolev estimates.
\end{proof}
\vskip 0.3cm
Now, following \cite{FLM}, Section 4.3, we observe that, if $\varepsilon_0$ has been taken small enough, the asymptotic solution ${\mathbf w}_N$ is 
${\mathcal O}(h^{\delta N}e^{-S/h})$ uniformly on the set,
$$
\{ {\rm dist}(x, {\bigcup_{\gamma\in G}}\gamma )\geq \varepsilon_0\}\cap \{ V_1\leq 2\varepsilon_0\}\cap \left(\bigcup_{\gamma\in G}{\mathcal N}(\gamma)\cup {\mathcal W}_N(\gamma)\right)
$$ Moreover, by (\ref{agmonfaible}), the same is true for ${\mathbf v}$ on $\{ {\rm dist}(x, {\displaystyle\bigcup_{\gamma\in G}}\gamma )\geq \varepsilon_0\}\cap \{ V_1\leq 2\varepsilon_0\}$. Therefore, using also Proposition \ref{patchvw}, we can paste together $e^{i\theta_0}{\mathbf v}$ and $h^{-n/4}{\mathbf w}_N$ in order to obtain a function ${\mathbf u}_N$ that satisfies to the properties of the following proposition (see also \cite{FLM}, Proposition 4.6):
\begin{proposition}\sl\label{globalbkw}
There exists a function ${\mathbf u}_N$, smooth on $\ddot O_N:=\{{\rm dist}(x,\ddot O)\}< 2(Nk)^{2/3}$, such that,
\begin{eqnarray*}
(P-\rho){\mathbf u}_N ={\mathcal O}(h^{\delta N}e^{-\Re \varphi_N/h});\\
\partial^{\alpha}{\mathbf u}_N ={\mathcal O}(h^{-m_\alpha} e^{-\Re \varphi_N/h}),
\end{eqnarray*}
uniformly on $\ddot O_N$, where $\widetilde\varphi_N$ is as in (\ref{reimphi}).
Moreover, ${\mathbf u}_N$ can be written as in (\ref{summarywkb}) in $\displaystyle\bigcup_{\gamma\in G}{\mathcal W}_N(\gamma) \cap \{ {\rm dist} (x, \ddot O)\geq (Nk)^{2/3}\}$ (with $\beta_1^{\ell,m}(x_\gamma, 0)\not=0$), while ${\mathbf u}_N$ is ${\mathcal O}(h^{\delta N}e^{-\Re \varphi_N/h})$ away from $\displaystyle\bigcup_{\gamma\in G}{\mathcal W}_N(\gamma) \cap \{ x\notin\ddot O\}$.
\end{proposition}

\section{Comparison between asymptotic and true solution}\label{compsol}
\subsection{A priori estimates}
In the same spirit as in \cite{FLM}, Theorem 2.2, we start with an a priori estimate for the resonant state of $P$. From now on, we denote by ${\mathbf u}$ the out-going solution of 
\be
\label{resoEq}
P{\mathbf u} = \rho_1{\mathbf u},
\ee
 normalized in the following way: we fix some analytic distorted space (also known, more recently, as a Perfectly Matched Layer), of the form,
\be
\label{PLM}
\widetilde\R^n_\theta := \{ x+i\theta F(x) \, ;\, x\in\R^n\},
\ee
where $F\in C^\infty (\R^n; \R^n)$, $F(x) = 0$ if $|x|\leq R_0$,  $F(x) =x$ for $|x|$  large enough, and where $\theta>0$ is sufficiently small, and may also tend to 0 with $h$, but not too rapidly (here, we take $\theta =h|\ln h| =k$). Then, by definition, the fact that $\rho_1$ is a resonance of $P$ means that Equation (\ref{resoEq}) admits a solution in $L^2(\widetilde\R^n_\theta)$, and here we take $ {\mathbf u}$ in such a way that,
\be
\label{normU}
\Vert{\mathbf u}\Vert_{L^2(\widetilde\R^n_\theta)} =1.
\ee
As before, $d$ stands for the Agmon distance associated with the pseudometric $\min(V_1,V_2)_+dx^2$, and we denote by $B_d(S):=\{ x\in\R^n)\, ;\, d(0,x)<S\}$ the corresponding open ball of radius $S=d(0, \partial\ddot O)$.
Then, we first have,
\begin{proposition}\sl 
\label{estprior1}
For any compact subset $K\subset \R^n$, there exists $N_K\geq 0$ such that,
$$
\Vert e^{s(x)/h}{\mathbf u}\Vert_{H^1(K)} ={\mathcal O}(h^{-N_K}),
$$
uniformly as $h\to 0$, where $s(x)=\varphi (x)$ if $x\in B_d(S)$ and $s(x)=S$ otherwise.
\end{proposition}
\begin{proof} The proof if very similar to that of \cite{FLM} Theorem 2.2, with the only difference that here we have to deal with pseudodifferentilal operators, forbidding us to use  Dirichlet realizations and non smooth weight functions. Instead, we modify $V_1$ in a way similar to (\ref{V1mod}), and we regularize the weights as in Proposition \ref{regulLip}.
\vskip 0.3cm
We consider a ($h$-dependent) cutoff function $\hat\chi$ such that,
$$
\hat\chi(x) = 1  \mbox{ if }  V_1(x)\geq 2k^{2/3}\,;\, \hat\chi(x) = 0 \mbox{ if } V_1(x)\leq k^{2/3}\, ;\, \partial^\alpha \hat\chi ={\mathcal O}(k^{-2|\alpha|/3}),
$$
and we set,
$$
\hat V_1:= \hat\chi V_1 + k^{2/3} (1-\hat\chi )\quad ; \quad \hat P_1:= -h^2\Delta + \hat V_1;
$$
\be
\hat P= 
\left(\begin{array}{cc}
\hat P_1 & 0\\
0 & P_2
\end{array}\right) + hR(x,hD_x).
\ee
We denote by $\hat E$ the first eigenvalue of $\hat P$, and by $\hat v$ its first normalized eigenfunction. Moreover, we consider the Agmon distance $\hat d$ associated with the pseudometric $\left( \min ( V_1,V_2)_+- \hat E\right)dx^2$, and we set $\hat\varphi (x):= \hat d(0,x)$. Then, the same proof as in \cite{FLM}, Lemma 3.1, shows the existence of a constant $C_1>0$ such that,
\be
s(x) - C_1 k\leq \hat \varphi (x) \leq \varphi (x)\quad (x\in\R^n).
\ee
Moreover, an adaptation of the proof of \cite{FLM}, Lemma 3.2 (obtained by using Proposition \ref{regulLip} in order to regularize the Lipschitz weight) gives,
\be
\label{decagmD}
\Vert e^{\hat\varphi /h}\hat v\Vert_{H^1(\R^n)} ={\mathcal O}(h^{-N_0}),
\ee
for some $N_0\geq 0$. Then, the result follows by considering the function $\hat\chi \hat v$, and by observing that, thanks to (\ref{decagmD}), one has (see \cite{FLM}, Lemma 3.3 and Formula (3.20)),
$$
\Vert \hat\chi \hat v - \frac1{2i\pi}\int_\gamma (z-P_\theta )^{-1}\hat\chi\hat v \,dz\Vert_{H^1} ={\mathcal O}(h^{-N_1}e^{-S/h}).
$$
Here, $\gamma$ is the oriented complex circle $\{ z\in\C\, ;\, |z-\hat E| = h^2\}$, and $P_\theta$ is a convenient distortion of $P$. The previous estimate actually shows that the distorted ${\mathbf u}_\theta$ of ${\mathbf u}$ coincides -- up to  ${\mathcal O}(h^{-N_1}e^{-S/h})$ -- with $\mu \hat\chi\hat v$, where $\mu$ is a complex constant satisfying $|\mu| = 1+{\mathcal O}(e^{-\delta /h})$, for some $\delta >0$.
\end{proof}
\begin{remark}\sl \label{rmkutheta} The previous proof also gives a global estimate on ${\mathbf u}_\theta$,
 $$
 \Vert e^{s(x)/h}{\mathbf u}_\theta \Vert_{H^1(\R^n)}={\mathcal O}(h^{-N_1'}),
 $$
  for some constant $N_1'\geq 0$. See \cite{FLM}, Lemma 3.3 and Formula (3.20).
\end{remark}
\vskip 0.5cm
Now, we plan to give an even better a priori estimate on the difference ${\mathbf u} - {\mathbf u}_N$ near the boundary of the island.
Here again, we follow the arguments given in \cite{FLM}, Section 5. For any $N\geq 1$, we set,
$$
U_N:= \{ x\in\R^n\, ;\, \dist (x, \partial\ddot O )< 2(Nk)^{2/3}\}.
$$
We have (see \cite{FLM}, Propositions 5.1 and 5.2),
\begin{proposition}\sl 
\label{dansUN}
There exists $N_1\geq 0$ and $C\geq 1$ such that, for any $N\geq 1$ large enough, one has,
$$
\Vert {\mathbf u} - {\mathbf u}_{CN}\Vert_{H^1(U_N)}\leq h^{-N_2}e^{-S/h}.
$$
\end{proposition}
\begin{proof}
We just recall the main lines of the proof in \cite{FLM}. At first, thanks to Proposition \ref{estprior1} and the particular form of ${\mathbf u}_{CN}$, we immediately see that the estimate is true on the set $\{\varphi (x)\geq S-2k\}$. Then, we take a cutoff function $\widetilde\chi\in C_0^\infty (\varphi (x)< S-k)$ such that $\widetilde\chi =1$ on $\{\varphi (x)\geq S-2k\}$, and $\partial^\alpha\widetilde\chi ={\mathcal O}(h^{-N_\alpha})$ for some $N_\alpha \geq 0$.
We also consider the Lipschitz weight,
$$
\phi_N(x)=\min\left(
\varphi (x)+C_1Nk +k(S-\varphi (x))^{1/3}, S+(1-k^{1/3})(S-\varphi (x)))
\right),
$$
and, by using Propositions \ref{estprior1} and \ref{globalbkw}, we see that if $C$ is large enough, we have,
$$
\Vert e^{\phi_N/h}(P-\rho_1)\widetilde \chi ({\mathbf u} - {\mathbf u}_{CN})\Vert_{L^2(\R^n)} ={\mathcal O}(h^{-M_1}),
$$
for some $M_1\geq 0$ independent of $N$. Then, regularizing $\phi_N$ as in Proposition \ref{regulLip}, we can perform Agmon estimates as in the proof of \cite{FLM}, Proposition 5.1, and we find,
$$
\Vert e^{\phi_N/h}\widetilde \chi ({\mathbf u} - {\mathbf u}_{CN})\Vert_{L^2(\R^n)} ={\mathcal O}(h^{-M_2}),
$$
for some $M_1\geq 0$ independent of $N$, and the result follows.
\end{proof}
\subsection{Propagation}
Now, we plan to prove (see \cite{FLM}, Proposition 6.1),
\begin{theorem}\sl
\label{estimfinale}
For any $L>0$ and for any $\alpha\in\Z_+^n$, there exists $N_{L,\alpha}\geq 1$ such that, for any $N\geq N_{L,\alpha}$, one has, 
\be
\label{vhl}
\partial_x^\alpha ({\mathbf u} - {\mathbf u}_{CN})(x,h)={\mathcal O}(h^Le^{-S/h})\quad {\rm as}\,\,\,\, h\to 0,
\ee
uniformly in $U_N$.
\end{theorem}
\begin{proof}
As in \cite{FLM}, the proof relies on three different types of microlocal propagation arguments. We fix some $\hat x\in\partial\ddot O$, and we define the Fourier-Bors-Iagolnitzer transform $T$ (see, e.g., \cite{Sj, Ma2}) as,
$$
T u(x,\xi ;h):= \int_{\R^n} e^{i(x-y)\xi /h -(x-y)^2/2h} u(y)dy.
$$
\vskip 0.3cm
1) {\it  Standard $C^\infty$ propagation}
\vskip 0.3cm
Since ${\mathbf u}$ is outgoing (that is, it becomes $L^2$ when restricted to distorted space or Perfectly Matched Layer defined in (\ref{PLM})), one can see as in \cite{FLM}, Lemma 6.2 that, if $t_0>0$ is large enough, then,  one has,
$$
T {\mathbf u}(x,\xi) ={\mathcal O}(h^\infty e^{-S/h}),
$$
uniformly near $\exp (-t_0H_{p_1})(\hat x, 0)$. Moreover, by Proposition \ref{estprior1}, we know that $e^{S/h}{\mathbf u}$ remains ${\mathcal O}(h^{-N_0})$ (for some $N_0\geq 0$) on a neighborhood of the $x$-projection of $\{ \exp (-tH_{p_1})(\hat x, 0)\, ;\, 0<t\leq t_0\}$.

Then, the standard $C^\infty$ propagation of the Frequency Set for the solution to a real principal type operator (see, e.g., \cite{Ma2}) shows that the previous estimate remain valid near $\exp (-tH_{p_1})(\hat x, 0)$ for any $t>0$.
\vskip 0.3cm
2) {\it Non standard propagation in $h$-dependent domains}
\vskip 0.3cm
Thanks to the previous result, we can concentrate our attention to a sufficiently small neighborhood of $\hat x$. As before, we choose local Euclidean coordinates $(x',x_n)\in \R^{n-1}\times \R$ centered at $\hat x$, such that $V_1(x) = -C_0x_n +{\mathcal O}(|x-\hat x|^2)$.
We also set $\mu_N:= (Nk)^{-1/3}$, and we considered the modified Fourier-Bors-Iagolnitzer transform $T_N$ defined by,
\be
\label{TN}
T_N u(x,\xi ;h):= \int_{\R^n} e^{i(x-y)\xi /h -(x'-y')^2/2h-\mu_N(x_n-y_n)^2/2h} u(y)dy.
\ee
Then, using the previous result it is elementary to show that (see \cite{FLM}, Lemma 6.3), for any (fixed) $t>0$ small enough, one has,
$$
T _N{\bf 1}_{K_1}{\mathbf u}(x,\xi) ={\mathcal O}(h^\infty e^{-S/h}),
$$
uniformly near $\exp (-tH_{p_1})(\hat x, 0)$,. Here $K_1$ is of the form $K_1 = K\backslash B_d(S)$, where $K$ is any compact neighborhood of the closure of $\ddot O$. The interest of the latter property is that, as shown in \cite{FLM}, it can be propagated up $h$-dependent times $t$ of order $(Nk)^{1/3}$. More precisely, setting,
$$
\exp tH_{p_1}(\hat x, 0) = (x'(t), x_n(t) ; \xi'(t), \xi_n(t))\quad (t\in\R),
$$
we have (see \cite{FLM}, Lemma 6.4),
\begin{lemma}\sl 
\label{propagation1}
There exists $\delta_0 >0$ such that, for any $\delta\in (0,\delta_0]$, for all $N\geq 1$ large enough, and for $t_{N,\delta}:=\delta^{-1}(Nk)^{1/3}$, one has,
$$
{\bf T}_N {\bf 1}_{K_1}{\mathbf u}={\mathcal O}(h^{\delta N}e^{-S/h} ) \mbox{ uniformly in } {\mathcal W}(t_N,h),
$$
where,
\begin{eqnarray*}
{\mathcal W}_\delta(N,h):= \{  |x_n- x_n(-t_{N,\delta})|\leq \delta (Nk)^{2/3}\, ,\, |\xi_n -\xi_n (-t_{N,\delta})| \leq \delta (Nk)^{1/3},\\
|x'- x'(-t_{N,\delta})|\leq \delta (Nk)^{1/3}\, ,\, |\xi' -\xi' (-t_{N,\delta})| \leq \delta (Nk)^{1/3}\}.
\end{eqnarray*}
\end{lemma}
\begin{proof}
The proof is based on the refined exponential weighted estimates (in the same spirit as in \cite{Ma2}) given in \cite{FLM}, Proposition 8.3 , that we apply here to the operator $P_1$. Since the proof is very similar to that of \cite{FLM}, Lemma 6.4, we omit the details.
\end{proof}
\vskip 0.3cm
On the other hand, using the explicit form of ${\mathbf u}_{CN}$ given in (\ref{summarywkb}), one also sees that, for any $L$ large enough, there exists $\delta_L>0$ such that, for any $N\geq 1$, one has (see \cite{FLM}, Lemma 6.7),
$$
T_N{\bf 1}_{K_1}{\mathbf u}_{CN}={\mathcal O}((h^{\delta_L N}+h^L)e^{-S/h} ) \mbox{ uniformly in } {\mathcal W}_\delta(N,h).
$$
In particular, taking $N=L/\delta_L$ with $L>>1$, we obtain a sequence $N=N_L$ along which,
$$
T_N{\bf 1}_{K_1}{\mathbf u}_{CN}={\mathcal O}(h^{\delta_L N}e^{-S/h} ) \mbox{ uniformly in } {\mathcal W}_\delta(N,h),
$$
and with both $N$ and $\delta_L N$ arbitrarily large.
\vskip 0.3cm
As a consequence, along the same sequence, we also obtain,
$$
T_N {\bf 1}_{K_1}({\mathbf u}-{\mathbf u}_{CN}) ={\mathcal O}(h^{\delta_L' N}e^{-S/h} ) \mbox{ uniformly in } {\mathcal W}_\delta(N,h),
$$
with $\delta'_L=\min(\delta, \delta_L)$.
\vskip 0.3cm
Moreover, we see that, when $y\in U_N\cap B_d(S)$ and $x\in \Pi_x {\mathcal W}_\delta(N,h)$ (where $\Pi_x$ stands for the natural projection onto the $x$-space), we have,
$$
\mu_N(x_n-y_n)^2 +s(x) - S \geq C_\delta Nk,
$$
with $C_\delta>0$ constant (and actually, $C_\delta\to \infty$ as $\delta\to 0$).
Therefore, using Proposition \ref{dansUN} and the expression (\ref{TN}) of $T_N$, we also obtain,
$$
T_N {\bf 1}_{U_N\cap B_d(S)}({\mathbf u}-{\mathbf u}_{CN}) ={\mathcal O}(h^{\delta N}e^{-S/h} ) \mbox{ uniformly in } {\mathcal W}_\delta(N,h).
$$
as a consequence, if we set,
\begin{equation}
\label{defchiN}
\chi_N (x):=\chi_0\left(\frac{|x_n-\hat x_n|}{(Nk)^{2/3}}\right)\chi_0\left(\frac{|x'-\hat x'|}{(Nk)^{1/2}}\right),
\end{equation}
where the function $\chi_0\in C_0^\infty (\R_+ ; [0,1])$ verifies $\chi_0 =1$ in a sufficiently large neighborhood of $0$, and is fixed  in such a way that $\chi_N(x) =1$ in $\{|x_n-\hat x_n| \leq |x_n(-t_N)-\hat x_n| + 2\delta (Nk)^{2/3}\,;\, |x'-\hat x'|\leq |x'(-t_N)-\hat x'|+2\delta (Nk)^{1/2}\}$ (here, $t_N$ and $\delta$ are those of Lemma \ref{propagation1}), then, the function,
$$
v_N:= \chi_N e^{S/h}({\mathbf u}-{\mathbf u}_{CN}),
$$
is such that,
\be
\label{estvN}
T_N v_N={\mathcal O}(h^{\delta_L' N}e^{-S/h} ) \mbox{ uniformly in } {\mathcal W}_\delta(N,h).
\ee
Moreover, we have (see \cite{FLM}, Section 6.2),
$$
(P-\rho_1) v_N =[P,\chi_N]e^{S/h}({\mathbf u}-{\mathbf u}_{CN})+{\mathcal O}(h^{\delta N}),
$$
and thus, on $\{d_N(x, \supp\nabla\chi_N) \geq \varepsilon\}\times \R^n$ (where $\varepsilon >0$ is fixed small enough and $d_N$ is  the distance associated with the metric $(Nk)^{-1}(dx')^2 + (Nk)^{-4/3}dx_n^2$), 
$$
T_N(P-\rho_1) v_N ={\mathcal O}(h^{\delta'N}),
$$
for some $\delta' =\delta'(\varepsilon) >0$.
\vskip 0.3cm
3) {\it (Almost) standard analytic propagation}
\vskip 0.3cm

Although we are in a region where no analytic assumption is made, a re-scaling of the problem makes appear estimates similar to those encountered in the analytic context. Indeed, setting,
$$
\widetilde h=\widetilde h_N:= \frac{h}{Nk} = \left( N\ln\frac1{h}\right)^{-1},
$$
and performing the change of variable (still working in the same coordinates for which $\hat x=0$),
\begin{eqnarray*}
&& x\mapsto 
\widetilde x = (\widetilde x', \widetilde x_n):= ((Nk)^{-1/2}x', (Nk)^{-2/3}x_n);\\
&& \xi\mapsto 
\widetilde \xi = (\widetilde\xi', \widetilde \xi_n):= ((Nk)^{-1/2}\xi', (Nk)^{-2/3}\xi_n)
\end{eqnarray*}
we see that the estimate (\ref{estvN}) implies (see \cite{FLM}, Formula (6.43)),
$$
T \widetilde v_N(\widetilde x, \widetilde\xi ; \widetilde h_N) ={\mathcal O}(e^{-\delta'_L/2\widetilde h_N})
$$
uniformly in the tubular domain,
\be
\begin{array}{rll}
\widetilde {\mathcal W}(\widetilde h):= \{  &|\widetilde x_n- \widetilde x_n(-\delta^{-1})|\leq \delta \, ,\, &|\widetilde \xi_n -\widetilde \xi_n (-\delta^{-1})| \leq \delta,\\
&|\widetilde x'- \widetilde x'(-\delta^{-1})|\leq \delta (Nk)^{-\frac 16}\, ,\, &|\widetilde \xi' -\widetilde \xi' (-\delta^{-1})| \leq \delta (Nk)^{-\frac 16}\},
\end{array}
\ee
where,
\begin{eqnarray*}
&& \widetilde v_N(\widetilde x) := (Nk)^{\frac{n-1}4 +\frac1{3}}v_N ((Nk)^{\frac12}\widetilde x',  (Nk)^{\frac23}\widetilde x_n);\\
&& (\widetilde x(\widetilde t), \widetilde\xi(\widetilde t)):=\exp \widetilde t H_{\widetilde p_1}(0,0); \\
&& \widetilde p_1(\widetilde x, \widetilde \xi) :=(Nk)^{1/3}| \widetilde\xi'|^2 + \widetilde\xi_n^2 +W_1(\widetilde x, \widetilde h);\\
&&  W_1(\widetilde x, \widetilde h):= (Nk)^{-2/3} V_1((Nk)^{1/2}\widetilde x', (Nk)^{2/3}\widetilde x_n) - (Nk)^{-2/3}\rho_1.
\end{eqnarray*}
Moreover, setting,
$$
\widetilde  P:= -(Nk)^{1/3}\widetilde h^2\Delta_{\widetilde  x'} - \widetilde h^2\partial_{\widetilde x_n}^2 + W_1(\widetilde  x),
$$
then, for any $N\geq 1$ large enough, we also have,
$$
T\widetilde  P\widetilde  v_N (\widetilde x, \widetilde\xi ; \widetilde h_N) ={\mathcal O}(e^{-\delta'/2\widetilde h_N}),
$$
uniformly with respectto $h>0$ small enough and $(\widetilde x, \widetilde\xi )\in\R^{2n}$ verifying $d_N(((Nk)^{1/2}\widetilde x', (Nk)^{2/3}\widetilde x_n) , \supp\nabla\chi_N) \geq \varepsilon$.
\vskip 0.3cm
Finally, by Proposition \ref{estprior1} and Proposition \ref{dansUN}, we have the a priori estimate,
$$
\Vert \widetilde v_N\Vert_{H^1} ={\mathcal O}(h^{-N_1}) ={\mathcal O}(e^{N_1/(N\widetilde h)}),
$$
for some $N_1\geq 0$ independent of $N$, and we observe that, for $N=L/\delta_L$, one has $N_1/(\delta_LN )\to 0$ as $L\to +\infty$.
\vskip 0.3cm
At this point, a small refinement of the propagation of the microsupport (see \cite{FLM}, Proposition 6.8) gives the existence of a constant $\delta_1>0$ independent of $L$ such that, for all $L$ large enough and $N=L/\delta_L$, one has,
\begin{equation}
\label{0notinMS}
T\widetilde v_N (\widetilde x,\widetilde\xi ;\widetilde h) ={\mathcal O}(e^{-\delta_1\delta_L /\widetilde h}),
\end{equation}
uniformly in $V(\delta_1)=\left\{\widetilde x;|\widetilde x|\leq \delta_1\right\}\times \left\{\widetilde \xi;\,\,(Nk)^{\frac 16}|\widetilde \xi'|+|\widetilde \xi_n|\leq \delta_1\right\}$.
\vskip 0.3cm
Then, using an ellipticity property of $\widetilde p_1$ away from $ \{\widetilde \xi;\,\,(Nk)^{\frac 16}|\widetilde \xi'|+|\widetilde \xi_n|\leq \delta_1\}$, and reconstructing $\widetilde v_N$ from $T\widetilde v_N$, one finally finds,
$$
\Vert \widetilde v_N\Vert_{H^m(|\widetilde x|\leq \delta_2)} ={\mathcal O}(e^{-\delta_2\delta_L/\widetilde h}),
$$
with $m\geq 0$ arbitrary, $\delta_2>0$ independent of $L$, $N=L/\delta_L$, and $L$ arbitrarily large. Therefore, turning back to the original coordinates $x$ and parameter $h$, and making $\hat x$ vary on all of $\partial\ddot O$, Theorem \ref{estimfinale} follows.
\end{proof}
\section{Asymptotics of the width}\label{secwidt}
As before, we denote by $P_\theta$ the distorted operator obtained from $P$ by means of a complex distortion as in (\ref{PLM}), with $R_0$ sufficiently large in order to have $\ddot O \subset \{ |x|\leq R_0/2\}$. We also denote by ${\mathbf u}_\theta$ the corresponding distorted state obtain from ${\mathbf u}$ by applying the same distortion (see, e.g., \cite{FLM} for more details).
\vskip 0.3cm
Let $\psi_0\in C_0^\infty ([0,2); [0,1])$ with $\psi_0 =1$ near $[0,1]$, and set,
$$
\psi_N (x):= \psi_0\left( \frac{\dist (x,\ddot O)}{(Nk)^{2/3}}\right),
$$
where, as before, $N=L/\delta_L$ with $L\geq 1$ arbitrarily large.
\vskip 0.3cm
Then, since $\psi_N{\mathbf u}=\psi_N{\mathbf u}_\theta$,  $P_\theta {\mathbf u}_\theta=\rho_1 {\mathbf u}_\theta$, and $\psi_N P_\theta \psi_N {\mathbf u}_\theta = \psi_N P \psi_N{\mathbf u}$, we have,
$$
\Im \rho_1 \Vert \psi_N{\mathbf u}\Vert ^2 =\Im \la \psi_N P_\theta{\mathbf u}_\theta, \psi_N{\mathbf u}\ra =\Im \la  [\psi_N, P_\theta] {\mathbf u}_\theta,\psi_N{\mathbf u}\ra,
$$
and thus, 
\be
\label{formres0}
\Im\rho_1=\Im \frac{\la 2h^2(\nabla\psi_N)\nabla {\mathbf u} +h^2(\Delta\psi_N){\mathbf u}, \psi_N {\mathbf u}\ra+h\la  [\psi_N,R_\theta]{\mathbf u}_\theta, \psi_N{\mathbf u}\ra}{\Vert \psi_N{\mathbf u}\Vert ^2}.
\ee

Moreover, we know that $\Vert\psi_N {\mathbf u}\Vert  =1+{\mathcal O}(e^{-\delta /h})$ with $\delta >0$ and, by Theorem \ref{estimfinale}, on $\supp\,\widetilde\psi_N$, we can replace ${\mathbf u}$ by ${\mathbf u}_{CN}$, up to an error ${\mathcal O}(h^L e^{-S/h})$. Also using Proposition \ref{estprior1}, we deduce,
\begin{eqnarray}
\label{formres1}
\Im\rho_1=\Im \la 2h^2(\nabla\psi_N)\nabla {\mathbf u}_{CN} +h^2(\Delta\psi_N){\mathbf u}, \psi_N {\mathbf u}_{CN}\ra\\
+h\la  [\psi_N,R_\theta]{\mathbf u}_\theta, \psi_N{\mathbf u}\ra\nonumber
+{\mathcal O}(h^{L-N_0})e^{-2S/h},
\end{eqnarray}
for some fix $N_0\geq 0$ independent of $L$.
\vskip 0.3cm
Now, we introduce $\widetilde\psi_0\in C_0^\infty ((1,2); [0,1])$ with $\widetilde \psi_0 =1$ near $\supp\nabla\psi_0$, and we set $\widetilde\psi_N(x) =\widetilde\psi_0\left( \frac{\dist (x,\ddot O)}{(Nk)^{2/3}}\right)$. 
\begin{lemma} One has,
\be
\label{reste}
 \la  [\psi_N,R_\theta]{\mathbf u}_\theta, \psi_N{\mathbf u}\ra=\la \psi_N [\psi_N,R]\widetilde\psi_N{\mathbf u}, \widetilde\psi_N{\mathbf u}\ra+{\mathcal O}(h^\infty e^{-2S/h}).
\ee
\end{lemma}
\begin{proof}
Thanks to Assumption 3, we can make in $ [\psi_N,R_\theta]$ the (complex) change of contour of integration,
$$
\R^n\ni\xi \mapsto \xi +i\sqrt{M_0}\frac{x-y}{\sqrt{(x-y)^2 + h^2}}.
$$
We obtain,
$$
 [\psi_N,R_\theta]{\mathbf u}_\theta (x)=\frac1{(2\pi h)^n}\int e^{i(x-y)\xi/h -\Phi /h})(\psi_N(x)-\psi_N(y))\widetilde r_\theta \overline{\mathbf u}_\theta (y)dyd\xi
$$
with,
$$
\Phi:= \sqrt{M_0}\frac{(x-y)^2}{\sqrt{(x-y)^2 + h^2}}\quad ; \quad \partial_{x,y}^\alpha\partial_\xi^\beta \widetilde r_\theta (x,y,\xi)={\mathcal O}(h^{-|\alpha| }\la\xi\ra).
$$
By construction, on the set,
$$
A_N:=\supp (\psi_N(x)-\psi_N(y))\cap\{ \widetilde\psi_N(x)\not=1 \mbox{ or } \widetilde\psi_N(y)\not=1\},
$$
we have $|x-y|\geq c(Nk)^{2/3}$ for some constant $c>0$. As a consequence, on this set, the quantity $|x-y|/\sqrt{(x-y)^2 + h^2}$ tends to 1 uniformly as $h\to 0$. Moreover, still on this set, we have either $s(x)=S$ or $s(y)=S$, and since $|s(x)-s(y)|\leq \mu|x-y|$ with $0<\mu <\sqrt{M_0}$, we deduce the existence of a constant $c_0>0$ such that,
$$
\mbox{ For } (x,y)\in A_N,\mbox{ one has } s(x)+s(y) +\Phi \geq 2S + c_0(Nk)^{2/3}.
$$
Therefore, by the Calder\'on-Vaillancourt theorem (and also using Proposition \ref{regulLip} in order to regularize the function $s(x)$), we obtain,
\begin{eqnarray*}
&&\Vert e^{-s/h}[\psi_N,R_\theta]e^{-s/h}(1-\widetilde\psi_N)\la hD_n\ra^{-1}\Vert \\
&& \hskip 2cm +\Vert (1-\widetilde\psi_N)e^{-s/h}[\psi_N,R_\theta]e^{-s/h}\la hD_n\ra^{-1}\Vert ={\mathcal O}(h^\infty e^{-2S/h}).
\end{eqnarray*}
Then, writing,
$$
 \la  [\psi_N,R_\theta]{\mathbf u}_\theta, \psi_N{\mathbf u}\ra= \la  e^{-s/h}[\psi_N,R_\theta]e^{-s/h}(e^{s/h}{\mathbf u}_\theta), \psi_Ne^{s/h}{\mathbf u}\ra,
$$
and using Proposition \ref{estprior1} and Remark \ref{rmkutheta}, the result follows.
\end{proof}
\vskip 0.3cm
Inserting (\ref{reste}) into (\ref{formres1}), and approaching $\widetilde\psi_N{\mathbf u}$ by $\widetilde\psi_N{\mathbf u}_{CN}$, we obtain,
\begin{eqnarray}
\label{resform2}
\Im\rho_1&=&\Im \la 2h^2(\nabla\psi_N)\nabla {\mathbf u}_{CN} +h^2(\Delta\psi_N){\mathbf u}, \psi_N {\mathbf u}_{CN}\ra\\
&& +h\la \psi_N [\psi_N,R]\widetilde \psi_N{\mathbf u}_{CN}, \widetilde\psi_N{\mathbf u}_{CN}\ra\nonumber
+{\mathcal O}(h^{L-N_0})e^{-2S/h}.
\end{eqnarray}
Finally, using Proposition \ref{globalbkw} (in particular the expression (\ref{summarywkb}) of ${\mathbf u}_{CN}$
 in $\displaystyle\bigcup_{\gamma\in G}{\mathcal W}_N(\gamma) \cap\supp\,\widetilde\psi_N$), we can perform a stationary-phase expansion in (\ref{resform2}) (as in \cite{FLM}, Section 7), and, for $L$ large enough, we obtain,
 $$
 \Im\rho_1 =-h^{(1-n_\Gamma)/2} \sum_{j=n_0}^{L}\sum_{0\leq m\leq\ell\leq L}f_{j,\ell, m}h^{j +\ell }|\ln h|^me^{-2S/h}+{\mathcal O}(h^{L/2})e^{-2S/h},
 $$
 with $f_{n_0, 0, 0}>0$. In particular, the result for $\rho_1$ follows. 
 \vskip 0.3cm
 The result for $\rho_j$, $j\geq 2$ can be done along the same lines, by using a representation of $\Im \rho_j$ analogous to (\ref{formres0}), and by approaching ${\mathbf u}$ by a linear combination of WKB expressions similar to ${\mathbf u}_{CN}$ (where the number of terms depends on the asymptotic multiplicity of the resonance: See \cite{HeSj2}, Section 10).
\bigskip

\appendix
\vskip 1cm

\bigskip


{}

\end{document}